\documentclass[aps,prx,twocolumn,footnoteinbib,superscriptaddress]{revtex4-2}

\usepackage[usenames,dvipsnames]{color}
\usepackage{graphicx}
\usepackage{epstopdf}
\usepackage{amsmath}
\usepackage{amsthm}
\usepackage{amsfonts}
\usepackage{amssymb}
\usepackage{bbold}
\usepackage{bm}
\usepackage{dsfont}
\usepackage{mathtools}
\usepackage{hyperref}
\usepackage{braket}
\usepackage{color}
\usepackage{multirow}
\usepackage{array}
\usepackage{tikz}
\usetikzlibrary{quantikz}

\def\>{\rangle}
\def\<{\langle}

\DeclareMathOperator*{\Id}{\mathds{1}}

\newcolumntype{C}[1]{>{\centering\arraybackslash}p{#1}}

\newcommand\numberthis{\addtocounter{equation}{1}\tag{\theequation}}

\begin{document}

\title{Complexity phase transitions in instantaneous quantum polynomial-time circuits}

\author{Chae-Yeun Park}
\affiliation{Institute for Theoretical Physics, University of Cologne, 50937 K{\"o}ln, Germany}
\affiliation{Xanadu, Toronto, ON, M5G 2C8, Canada}

\author{Michael J. Kastoryano}
\affiliation{Institute for Theoretical Physics, University of Cologne, 50937 K{\"o}ln, Germany}
\affiliation{Amazon Quantum Solutions Lab, Seattle, Washington 98170, USA}
\affiliation{AWS Center for Quantum Computing, Pasadena, California 91125, USA}

\newtheorem{theorem}{Theorem}
\newtheorem{proposition}{Proposition}
\newtheorem{observation}{Observation}
\newtheorem{corollary}{Corollary}
\newtheorem{lemma}{Lemma}
\newtheorem{conjecture}{Conjecture}

\theoremstyle{definition}
\newtheorem{example}{Example}
\newtheorem{definition}{Definition}

\date{\today}

\begin{abstract}
We study a subclass of the Instantaneous Quantum Polynomial-time (IQP) circuit with a varying density of two-qubit gates.
In addition to a known anticoncentration regime,
we identify novel parameter conditions where the model is classically simulable or the output distribution follows the Porter-Thomas distribution.
By showing that those parameter regimes do not coincide, we argue the presence of more than two phases in the model. 
The learnability of the output distribution of this model is further studied, which indicates that an energy-based model fails to learn the output distribution even when it is not anticoncentrated.
Our study reveals that a quantum circuit model can have multiple fine-grained complexity phases, suggesting the potential for quantum advantage even when the output distribution is far from the Porter-Thomas distribution.
\end{abstract}
\maketitle

It is widely believed that classical computers require an exponential amount of resources to simulate a general quantum system.
Computational complexity theory provides a rigorous mathematical foundation of this belief 
by proving that the output distribution of a certain subclass of quantum system is classically intractable. 
Such examples include BosonSampling~\cite{aaronson2011computational}, instantaneous quantum polynomial-time (IQP) circuits~\cite{bremner2011classical,bremner2016average} as well as random circuits with one and two qubit gates~\cite{boixo2018characterizing,bouland2019complexity} on which Google's recent ``quantum advantage'' claim~\cite{arute2019quantum} is based.
On the other hand, when the unitary operators allowed in the system are insufficient for universal quantum computation, one can often devise an efficient classical algorithm for this task.
Such examples include free fermions~\cite{lieb1961two}, matchgates~\cite{knill2001fermionic,valiant2002quantum,terhal2002classical}, and Clifford~\cite{gottesman1998heisenberg,aaronson2004improved} circuits.
Then a natural question arises: ``Where is the boundary between classically simulable and intractable quantum systems?"

For Clifford circuits, it is known that adding any non-zero density of $T$-gates induces a transition to a classically intractable circuit~\cite{zhou2020single,haferkamp2020quantum,bravyi2016improved,leone2021quantum}.
However, only a partial answer is available for other models.
Several studies~\cite{deshpande2018dynamical,oh2022classical,maskara2022complexity} showed that classically easy and difficult phases exist in dynamics of interacting bosons, but those results are often limited to the worst-case instances.
On the other hand, Ref.~\cite{tangpanitanon2019quantum,tangpanitanon2023signatures} claimed that a transition between thermalizing and many-body localized phases in interacting spins with a periodic driving is related to  sampling complexity.
Albeit their numerical results demonstrate that the output distribution follows the Porter-Thomas distribution in the thermalizing phase, 
their complexity arguments rely on a strong assumption that circular unitary ensemble is classically intractable.

In this work, we study complexity phase transitions in IQP circuits using the density of two-qubit gates as a parameter.
Compared to other models, rigorous arguments on the classical hardness of average instances have been established for the IQP circuits~\cite{bremner2011classical,bremner2016average,bremner2017achieving}.
In addition to a known parameter condition for anticoncentration~\cite{bremner2017achieving}, we find parameter regimes where the model is classically simulable or the output distribution follows the Porter-Thomas distribution.
By showing that those parameters regimes do not coincide, we suggest the model has more than two distinct phases.

We further explore the classical learnability of the output distribution.
We first show that the output distribution of the worst-case instance is not PAC-learnable for circuits with $N^{\Omega(1)}$ two-qubit gates.
However, this theoretical tool is too limited when studying typical instances.
We address this problem by introducing a tractability measure based on the classical parent Hamiltonian whose Gibbs state represents the distribution.
Our measure is closely related to the energy-based model~\cite{lecun2006tutorial} that has been widely used for modeling the distribution.
By computing the tractability measure as well as training a neural network, we argue that classical learnability also undergoes a transition but earlier than the transition to the anticoncentration or to the Porter-Thomas distribution.
This further supports that the model has several complexity phases.

\textit{Instantaneous quantum polynomial-time circuit in the Erd\H{o}s-R\'{e}nyi graph}.--
Throughout the paper, we consider the output distribution $p(x) = |\braket{x|\psi}|^2$ of a quantum state generated by an IQP-type circuit $\ket{\psi} = H^{\otimes N} U_{\rm diag} \ket{+}^{\otimes N}$ where $U_{\rm diag}$ is a diagonal unitary gate in the computational basis, $H$ is the Hadamard gate, $N$ is the number of qubits, and $\ket{x}$ is a vector in the computational basis.
We especially use 
\begin{align}
U_{\rm diag} = F_{\theta_i, \phi_{ij} }=\exp[i \sum_i \theta_i Z_i] \exp[i \sum_{i < j} \phi_{ij} Z_i Z_j]
\end{align}
where $Z_i$ is the Pauli-Z matrix acting on the $i$-th qubit.

We always sample $\{\theta_i\}$ from the uniform distribution $\mathcal{U}_{[0, 2\pi]}$ but two body interactions $\{\phi_{i,j}\}$ are only turned on with a probability $q$, i.e., for each pair $(i,j)$, we set $\phi_{ij} = \Theta(q-\tilde{q}) \tilde{\phi}_{ij}$ where $\Theta(x)$ is the Heavyside step function, $\tilde{q}$ is a random sample from $\mathcal{U}_{[0,1]}$, and $\tilde{\phi}_{ij}$ is randomly sampled from $\mathcal{U}_{[0, 2\pi]}$.
In other words, the two-body interaction terms are chosen using the Erd\H{o}s-R\'{e}nyi random graph~\cite{erdHos1960evolution}.
Formally, the Erd\H{o}s-R\'{e}nyi $G(N,q)$ is defined for $N$ vertices, and each unique pair of vertices $(i, j)$ are connected by an edge with probability $q$.
Thus the diagonal gate contains only single-qubit gates when $q=0$, whereas it also has $2$-qubit gates when $q>0$.
As the Erd\H{o}s-R\'{e}nyi random graph shows phase transitions in its connectedness~\cite{erdHos1960evolution} and treewidth~\cite{lee2012rank},
it is natural to expect that IQP circuits over the random graph also show several phases in their properties.

The hardness proof of the sampling task typically relies on three conjectures: (C1) The polynomial hierarchy is infinite, (C2) It is \textsf{\#P-Hard} to approximate (up to a constant relative error) the output probability $p(x)$ for a constant fraction of instances, and (C3) the output distribution is anticoncentrated.
For the IQP circuit over a complete graph (i.e., $q=1$), Ref.~\cite{bremner2016average} showed that the output distribution satisfies C3, while C2 is true if approximating the partition function of the classical Ising model with complex weights defined over the same graph is \textsf{\#P-Hard} for a constant fraction of instances (which is widely believed to be true).
The result is relaxed to a sparse graph with $q=\mathcal{O}(\log N/N)$ in Ref.~\cite{bremner2017achieving}, but it requires a stronger condition that computing the partition function of the Ising model in a sparse graph is also difficult, which is less obvious.
While approximating the output probability of IQP circuits over some 2D lattices has been proven to be \textsf{\#P-Hard}~\cite{gao2017quantum,bermejo2018architectures,haferkamp2020closing}, the average-case hardness for circuits over random graphs remains as a conjecture.
Thus, the primary aim of the first part of this paper is to unveil the relationship between graph sparsity and the complexity of the sampling task with several theoretical and numerical techniques.

\textit{Sampling difficulty of the IQP circuit.}--
Our first result is that there is a classical algorithm for sampling from IQP circuits over the random graph when it is too sparse.
\begin{proposition}[Classical algorithm for sampling]
    Let $C$ be the distribution of IQP circuits over the Erd\H{o}s-R\'{e}nyi random graph with the connection probability $q < 1/N$. Then $1-\mathcal{O}(1/N)$ fraction of $C$ allows an efficient classical algorithm for sampling from the output distribution.
\end{proposition}
We briefly sketch the idea here. Full details of the algorithm can be found in Appendix~\ref{app:classical_algo_sampling}.
For a given graph $G$, our circuit has a two-qubit gate between $i$-th and $j$-th qubits iff the graph has an edge $(i,j)$.
We can interpret the circuit as a tensor network, which becomes a subgraph of the line graph of $G$.
Then the tensor network contraction algorithm~\cite{markov2008simulating} can simulate the circuit in $\mathcal{O}(N^3 \exp(\Delta(G) \mathrm{tw}(G)))$ time where $\Delta(G)$ is the maximum degree of $G$ and $\mathrm{tw}(G)$ is the treewidth of $G$.
When $q < 1/N$, the maximum degree of $G$ is less than $\mathcal{O}(\log N)$ for $1-\mathcal{O}(1/N)$ fraction of instances and the treewidth of the random graph is constant~\cite{lee2012rank}, so we have a polynomial time algorithm for computing the output distribution. One can also extend this algorithm for the sampling task using the tensor network description of the marginal distribution instead of the output state.

\begin{figure}
    \centering
    \includegraphics[width=0.90\linewidth]{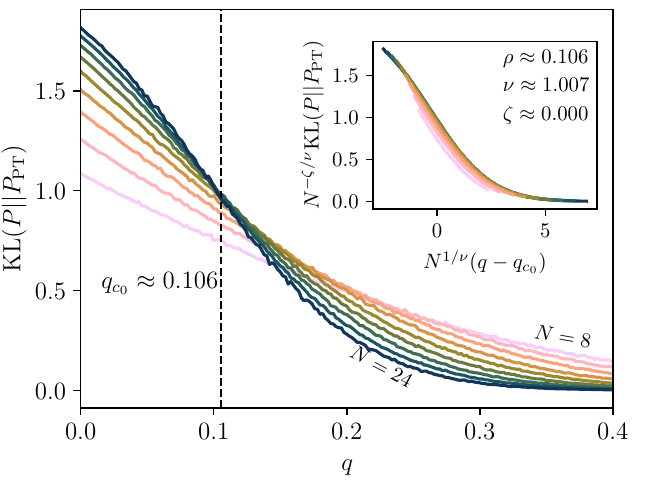}
    \caption{Phase transition of the KL divergence between the observed distribution of the output probability $p(x)$ from the IQP circuit and the Porter-Thomas distribution. Data from the system of sizes $N=8$ to $N=24$ collapses to a single curve. The transition point $p_{c} \approx 0.106$ is observed. The results are obtained from $2^{14}$ random circuit instances for $8 \leq N\leq 22$ and $2^{10}$ instances for $N=24$. For each instance of a random circuit, we randomly sampled all parameters and graph edges. As datapoints for small $N$ are subject to discretization error (as we binned the distributions), they are excluded from the data collapse analysis.}
    \label{fig:phase_transition_kl_div}
\end{figure}

We next show that the output distribution of IQP circuits over the random graph is anticoncentrated (C3).
Formally, we say that the output distribution is anticoncentrated if there exists constants $\alpha, \beta >0$ such that for any $x \in \{0,1\}^{N}$,
\begin{align}
    \mathrm{Pr}_{U} \Bigl( |\braket{x|U|0}|^2 \geq \frac{\alpha}{2^N} \Bigr) \geq \beta.
\end{align}
The following proposition shows that IQP circuits over the random graph shows this property.

\begin{proposition}[Anticoncentration]
Consider IQP circuits defined over the Erd\H{o}s-R\'{e}nyi random graph. There exists a constant $\gamma > 0$ such that the output distribution of the circuit is anticoncentrated when $q \geq \gamma \log(N)/N$.
\end{proposition}

The anticoncentration of IQP circuits over the random graph but with diagonal gates of discrete parameters is proved in Ref.~\cite{bremner2017achieving}. We show in Appendix~\ref{app:anticoncentration} that the same result also holds for continuous random parameters as in our case.

Still, anticoncentration itself is not sufficient for the difficulty of the distribution.
Thus lots of existing literature in quantum sampling advantage~\cite{boixo2018characterizing,bermejo2018architectures,bouland2019complexity,arute2019quantum,madsen2022quantum} also have considered the exact distribution of the output bit-strings.
When the output distribution of a circuit shows quantum advantage, one expects that the distribution of the output probability $p(x)$ is close to that of a random vector in the Hilbert space, which is known as the Porter-Thomas~\cite{porter1956fluctuations} distribution.
Thus distance between the output distribution of the circuit and the Porter-Thomas distribution is widely used to verify quantum advantage in the sampling task~\cite{boixo2018characterizing,bermejo2018architectures,bouland2019complexity,arute2019quantum,madsen2022quantum}.

In our model, we numerically locate the critical point $q_{c}$ where the output distribution becomes the Porter-Thomas distribution for $q > q_{c}$.
\begin{observation}
    The output distribution follows the Porter-Thomas distribution in the thermodynamic limit ($N \gg 1$) if and only if $q > q_{c}$ with $q_{c} \approx 0.106$.
\end{observation}

Our observation is based on the strong finite-size scaling behavior of the Kullback–Leibler (KL) divergence between the statistics of output probabilities $p(x)$ and the Porter-Thomas distribution, which is shown in Fig.~\ref{fig:phase_transition_kl_div}.
Precisely, all data from 8 to 24 qubits collapse to a single curve $\tilde{f}$ after rescaling, which is given by
\begin{align}
    N^{-\zeta / \nu} y_N = \tilde{f}(N^{1/\nu}(q - q_{c})),
\end{align}
where $y_N$ are obtained data $\mathrm{KL}(P || P_{PT})$ for $N$-qubit system, $\nu, \xi$ are critical exponents, and $q_{c}$ is the phase transition point. 
In the thermodynamic limit ($N \gg 1$), our collapsing result implies that there is an abrupt change of the distribution at $q=q_{c}$.

We note that, even though our observation is based on numerical results, finite-size scaling analysis~\cite{fisher1972scaling,cardy_2012} is one of the most rigorous tools for identifying phase transitions and successfully used for analyzing phases in complex models~\cite{price2012critical,luitz2015many,skinner2019measurement,lavasani2021measurement}, whose theoretical investigations are limited.
In Appendix~\ref{app:anticon_vs_pt}, we show that there is a constant $\kappa > 0$ such that, for $q=\kappa (\log N / N)$, the KL divergence converges to a constant as $N$ increases albeit the distribution is anticoncentrated.
We also provide additional evidence of the phase transitions using entanglement spectra of output quantum states in Appendix~\ref{app:ent_iqp}.

Our results so far indicate that several complexity measures do not coincide, and the model may exhibit several phase transitions.
Nonetheless, we still do not know exactly when the distribution becomes classical intractable, 
as the complexity results rely on the assumption that approximating the partition function of the classical Ising model over corresponding graphs is also difficult~\cite{bremner2016average,bremner2017achieving}.
Therefore, proving computational hardness of the Ising model yields different consequences of our results.
For example, if computing the partition function of the Ising model over a sparse graph with $q=\mathcal{O}(\log N /N)$ is proven to be classically hard, there is a circuit whose output distribution does not follow the Porter-Thomas distribution but still classically difficult to be sampled from.
So there can be a better measure detecting sampling complexity of the distribution.
On the other hand, if computing the partition function with $q=\mathcal{O}(\log N /N)$ turns out to be easy (i.e., there is a classical algorithm for this task), there is a family of the IQP circuit whose output distribution is anticoncentrated but can be classically estimated. This can be an example of non-trivial classically tractable circuit with large entanglement (such as a Clifford circuit).

\textit{Learning the output distribution.--}
We now shift our focus to a learning task, i.e., whether a classical algorithm can reproduce the output distribution of a circuit when samples are provided.
Albeit the relation between sampling difficulty and learnability is subtle,
it is often believed that learning the output distribution of a circuit is also classically hard when the sampling task is difficult, as it implies that one can mimic quantum randomness~\cite{aaronson2023certified} using a classical computer with a polynomial number of samples from a real quantum device.
Hence, the ability of a classical machine learning algorithm to learn the distribution can be another diagnostic tool to measure the intrinsic complexity of the distribution.

As in the sampling task, we use a tool from the complexity theory, the PAC learning framework, to study the basic property of the circuit.
The PAC learning theory of discrete distributions often considers two different classes of classical models named \textit{evaluator} and \textit{generator}~\cite{kearns1994learnability}. 
Given a polynomial number of samples from $p(x)$, an evaluator model learns how to evaluate the probability distribution, i.e., for any given $x'$, it can evaluate $p(x')$ up to some error.
On the other hand, a generator model can generate a sample using random numbers.
Then the efficient learnability w.r.t. those models is defined as follows.

\begin{definition}
Let $\mathcal{D}_N$ be a class of distributions over $\{0, 1\}^N$. We say that $\mathcal{D}_N$ is efficiently learnable w.r.t. evaluator (or generator) if there is an algorithm that, for any target distribution $P \in \mathcal{D}_N$ and any given $\epsilon, \delta \in [0,1]$, outputs a circuit for evaluating $Q(x)$ (or generating samples from $Q$) with probability at least $1-\delta$ in time polynomial in $1/\epsilon$, $1/\delta$ and $N$, where the distribution $Q$ satisfies $\mathrm{TV}(P,Q) \leq \epsilon$. Here, $\mathrm{TV}(P,Q) = \sum_{x \in \{0, 1\}^N } |P(x) - Q(x)|/2$ is the total variational distance.
\end{definition}

We now consider the output distribution from IQP circuits.
When $q=0$, the circuit produces a product state, and we expect that it is easy to learn the distribution, which we prove as follows.
\begin{proposition}\label{prop:learnability_product}
    The output distribution from a product state is efficiently learnable both w.r.t. evaluator and generator.
\end{proposition}

On the other hand, under a plausible assumption, it is also possible to argue that learning the output distribution of IQP circuits over a sparse graph with any non-zero $q$ is classically hard for the worst-case instance.

\begin{proposition}\label{prop:larning_hardness_lpn}
    Consider the noisy parity distribution for $k+1$ bits $P_{s, \eta}(x, y)$ defined for $s,x \in \{0,1\}^k$, $\eta \in (0,1/2)$, and $y \in \{0,1\}$, which is given by $P_{s,\eta}(x, x \cdot s) = (1-\eta)2^{-k}$ and $P_{s,\eta}(x, \neg x \cdot s) = \eta 2^{-k}$. Under the assumption that a product distribution $P_{s, \eta} \otimes \Id^{k}$ is not efficiently learnable w.r.t. an evaluator, where $\Id^{k}$ is the uniform distribution over $k$ bits, the output distribution IQP circuits with $N^{\Omega(1)}$ two-qubit gates is not efficiently learnable w.r.t. an evaluator, either.
\end{proposition}

We provide proofs of Propositions~\ref{prop:learnability_product} and \ref{prop:larning_hardness_lpn} in Appendix~\ref{app:learning_output_dist_iqp}.

However, those results tell little about learning the \textit{average}-cases of the IQP circuit. We thus instead consider a concrete model for the probability distribution known as the energy-based model~\cite{lecun2006tutorial} and ask whether such a model can learn the distribution (this type of question is often called \textit{representation dependent} hardness)~\footnote{In principle, the energy-based model itself is neither an evaluator nor a generator, as the RBM contains the spin-glass problem evaluating probability or sampling from which is generally NP-Hard~\cite{barahona1982computational}.
In practical use case, however, where the Markov chain Monte-Carlo generates an accurate sample in a polynomial time, the probability distribution can also be estimated~\cite{jerrum1993polynomial,neal2001annealed}.
Thus it is both an evaluator and generator when this condition holds.}.

\begin{figure}[t]
    \centering
    \includegraphics[width=0.9\linewidth]{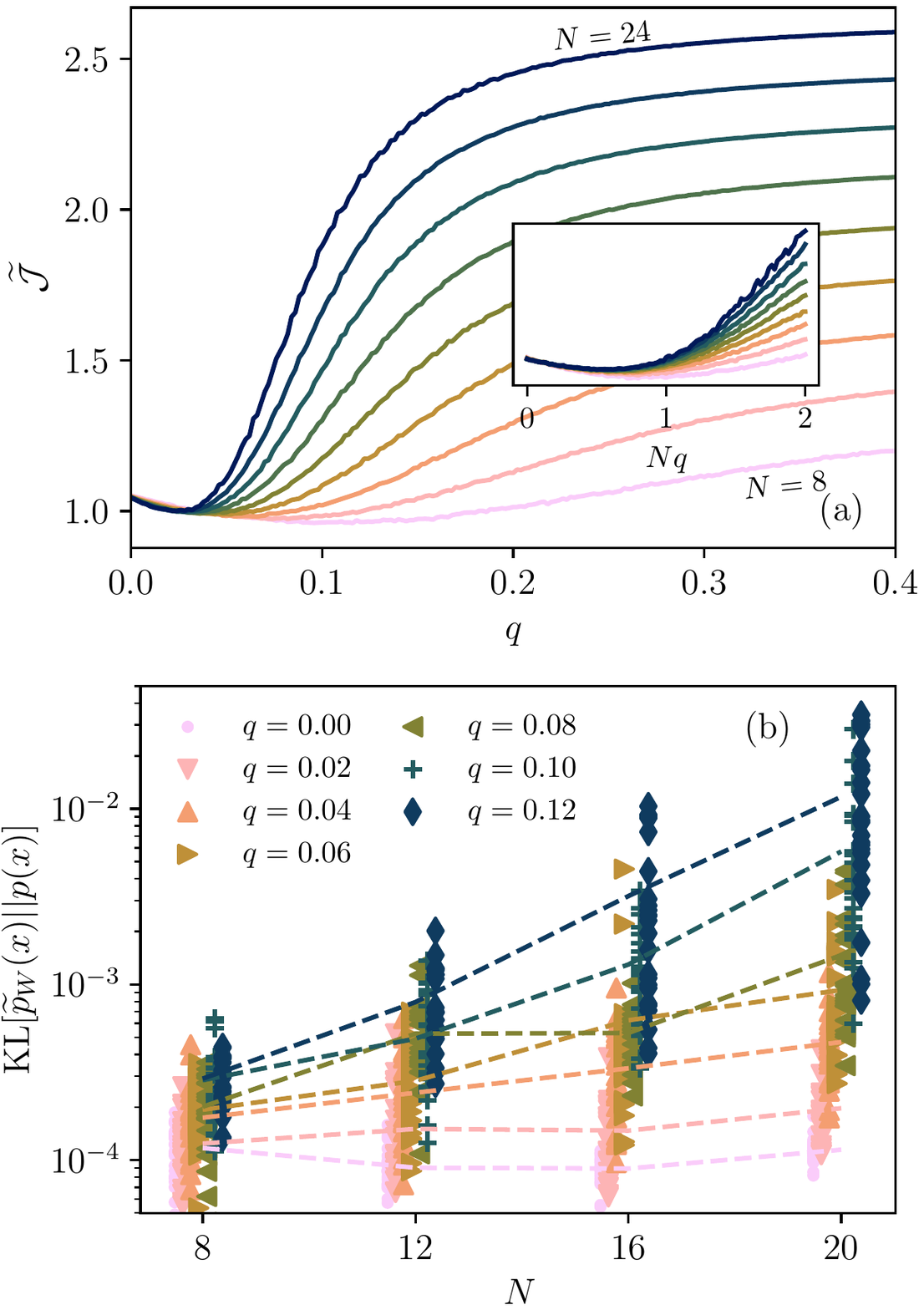}
    \caption{(a) Averaged sums of the Hamiltonian coefficients $\mathcal{J}=\sum_{k \neq 0} |J(k)|$ after normalization. (Inset) The same plot as a function of $Nq$. When $Nq \gtrsim 1.0$, $\widetilde{\mathcal{J}}=\log \mathcal{J}/\log N$ diverges with $N$, which implies a super-polynomial scaling of $\mathcal{J}$. (b) The Kullback-Leibler divergence between the output distribution of the IQP circuits $p(x)$ and that from the energy-based model after convergence $\tilde{p}_{W}(x)$. The parameter $q$ ranges in value from $q=0.0$ (the lightest color) to $q=0.12$ (the darkest color) with spacing $0.02$.
    Each data point shows the KL divergence between the output distribution from each circuit realization (randomly chosen $\{\theta_i\}$ and $\{\phi_{ij}\}$) and the trained distribution using that distribution.
    For each $q$ and $N$, we ran $24$ independent circuit realizations and show averaged results by lines.}
    \label{fig:learning_phases}
\end{figure}

When $p(x) \gneq 0$ for all $x$, one can consider a classical parent Hamiltonian which gives $p(x) \propto e^{-H(x)}$. An energy-based model learns $H(x)$ instead of $p(x)$.
As $x \in \{0,1\}^N$, one can always expand $H(x)$ as 
\begin{align}
    -H(x) &= J_0 + \sum_{i} J_i Z_i(x) + \sum_{i j} J_{ij} Z_i(x) Z_j(x) \nonumber \\
    &\qquad + \sum_{ijk}  J_{ijk} Z_i(x) Z_j(x) Z_k(x) + \cdots,
\end{align}
where the sum is over the distinct indices, and $Z_i(x) = (-1)^{x_i}$. 
This is just a change of variable $x$ to $1-2x =: z \in \{-1,1\}^N$.
Each energy-based model can learn different form of coefficients $\{J_0, J_i, J_{ij}, \cdots \}$.
For example, the Boltzmann machine~\cite{hinton2007boltzmann} (based on the classical Ising model) exactly uses the second order Hamiltonian (with $J_{i,j}$s).
Thus it can exactly capture any second order correlations of $x_i$, but not the higher order terms.
On the other hand, a single layer fully connected network with weights $w_i$ and $x^3$ activation function gives $-H(x) = (\sum_i w_i x_i)^3$, which can approximate some third order correlations.
However, any of those models cannot exactly express all high-order terms unless exponentially many parameters are provided, as the number of coefficients ($J$s) is exponentially large.

Based on this observation, we use two different diagnostics for the learning complexity of the output distribution. First, we compute $\mathcal{J} = \sum_{S \neq \emptyset} |J_S|$ where $J_S$ are coefficients from the Hamiltonian expansion $H(x) = -\sum_S J_S \prod_{i} Z_i(x)$, and the summation is over all subsets of indices $S \subset \{1, \cdots, N\}$.
For a given $p(x) \gneq 0$, this measure can be computed using the Hadamard-Walsh transformation of $\log p(x)$ (see Appendix~\ref{app:classical_ham}).
We also note that this measure can be seen as a classical analogy of the quantum circuit complexity~\cite{nielsen2005geometric,nielsen2006quantum,stanford2014complexity}.
Second, we indeed train a simple energy-based model for the output distribution of the IQP circuits.
For each output distribution of an IQP circuit $p(x)$ from a given graph and parameters $\{\theta_i, \phi_{ij}\}$, we optimize the log-likelihood cost function given as $\sum_x p(x) \log \tilde{p}_W(x) $ where $\tilde{p}_W(x)$ is the distribution of the energy-based model (i.e., $\widetilde{p}_W(x) \propto e^{f_W(x)}$ where $f_W(x)$ is the output of the network with weights $W$ for an input $x$). Then we compute $\mathrm{KL}[\tilde{p}_{W}(x) || p(x)]$, the distance between the converged distribution $\tilde{p}_{W}(x)$ and $p(x)$ (see Appendix~\ref{app:ml_setup} for detailed machine learning set-ups).

We can already expect some behaviors of those diagnostics for extreme cases.
When $q=0$, one can find a Hamiltonian with single body terms that exactly represents the output distribution, i.e., $p(x) \propto e^{-\sum_i J_i Z_i(x)}$ (see Appendix~\ref{app:classical_ham}).
In this case, $\mathcal{J} = \mathcal{O}(N)$ and we expect that a simple network sufficiently learns the distribution.
On the other hand, when the distribution is classically hard, we expect that our measure $\mathcal{J}$ is exponentially large and the model does not learn the distribution.
We thus introduce a normalized measure $\widetilde{J} = \log \mathcal{J} / \log N$ that converges to a constant when $\mathcal{J}$ is polynomial in $N$, whereas it diverges if $\mathcal{J}$ is exponentially large.

We plot the scaling of the measure $\widetilde{J}$ computed for the output distribution of the circuits in Fig.~\ref{fig:learning_phases}(a). Even though the exact phase transition point is not located, we see that the normalized measure $\widetilde{J}$ increases with $N$ for $Nq \gtrsim 1.0$, which suggests that learning the output distribution becomes difficult in this regime. 
We also plot Fig.~\ref{fig:learning_phases}(b) the distance between the distribution of the network after training $\tilde{p}_{W}(x)$ and the circuit output distribution $p(x)$ trained for up to $N=20$.

We see that the converged distance increases exponentially with $N$ even when $0.06 \lesssim q$, where the output distribution is still far from the scrambled Porter-Thomas distribution.
Thus both figures suggest that learning the distribution becomes difficult earlier than the distribution becomes anticoncentrated or scrambled.

In Appendix~\ref{app:classical_ham}, we provide additional numerical results supporting the learnability transition from \textit{approximate} parent Hamiltonians as well as machine learning results with a different optimization algorithm.
We further note that the networks we used here (with $\sim \alpha^2N^2$ parameters with $\alpha =30$) have $\sim 2 \times 10^3$ times more parameters than the IQP circuit itself (with $\sim N^2/2$ parameters), which agrees with observations that classical description of a distribution from a quantum system can be extremely inefficient~\cite{gao2017efficient,niu2020learnability,park2022expressive}.
Our results suggest that a classical learner cannot easily spoof the output distribution of a circuit even when it is far from the Porter-Thomas distribution.

\textit{Conclusion.--}
By combining the complexity arguments with extensive numerical studies, we have shown that IQP circuits over the random graph undergo several phase transitions related to complexity as the density of two-qubit gates increases.
We additionally argued that the classical learnability the output distribution also becomes classically intractable as the density increase but much earlier than that of the sampling.
Our study shows that the IQP circuit has several fine-grained complexity phases and suggests that the model may show a quantum advantage even when the output distribution is far from the Porter-Thomas distribution.
An interesting open question is whether models with 2D architecture~\cite{bermejo2018architectures,arute2019quantum,napp2022efficient,dalzell2022random} also have such an intermediate phase. As numerical study of those models require more careful analysis due to strong finite-size effects, we reserve the answer to this question for future research.

\textit{Acknowledgements}.--
CYP thanks Jaeyoon Cho, Changhun Oh, Youngrong Lim, and Anuj Apte for helpful discussions and David Wierichs for useful comments.
This project was funded by the Deutsche Forschungsgemeinschaft under Germany’s Excellence
Strategy - Cluster of Excellence Matter and Light for Quantum Computing (ML4Q) EXC 2004/1-390534769
and within the CRC network TR 183 (project grant 277101999) as part of project B01. The numerical simulations were performed on the Juwels and Juwels-Booster clusters at the Forschungszentrum Juelich. The work presented here was completed while both authors were at the University of Cologne.


%
\appendix

\onecolumngrid

\renewcommand{\thefigure}{A\arabic{figure}}
\setcounter{figure}{0}

\section{Classical algorithm for sampling from the IQP distribution} \label{app:classical_algo_sampling}
In this section, we (1) introduce a classical algorithm for sampling from the output distribution of IQP circuits and (2) show that more than $1-\mathcal{O}(1/N)$ fraction of the IQP circuit instances over the Erd\H{o}s-Reyni is efficiently classically simulable when the connection probability is smaller than $1/N$.
Our algorithm is based on tensor network contraction, whose complexity is given by in terms of the treewidth of the underlying graph~\cite{markov2008simulating}.

\subsection{Complexity of contracting tensor networks}
In this subsection, we introduce a result by Markov and Shi~\cite{markov2008simulating} which studied the complexity of contracting tensor networks.
Let us first define a tree decomposition of a graph.
\begin{definition}
For a graph $G=(V,E)$, a \textit{tree decomposition} of $G$ is a pair $(\mathcal{X}, T)$ where $\mathcal{X}=\{X_1,\cdots,X_n\}$ is a collection of vertices of $G$ (sometimes called bags), and $T$ is a tree whose nodes are $X_i$, which satisfies the following properties:
\begin{enumerate}
    \item For all $v \in V$, $v$ appears at least one bag.
    \item For every edge $(v,w) \in E$, there is a node in $T$ that contains both $v$ and $w$.
    \item If $X_i$ and $X_j$ both contain a vertex $v \in V$, all nodes $X_k$ of a tree in the path between $X_i$ and $X_j$ also contain $v$.
\end{enumerate}
\end{definition}

For each resulting tree from a tree decomposition, the \textit{width} of the tree is the size of largest bag minus 1.
The treewidth of a graph is defined as follows:
\begin{definition}
    The \textit{treewidth} of a graph $tw(G)$ is defined by the the minimum width among all possible tree decompositions of $G$.
\end{definition}

We next define the line graph as follows:
\begin{definition}
The line graph $L(G)$ of a undirected graph $G=(V,E)$ is a graph with vertex set $E$ where two vertices are adjacent if they have a common vertex in $G$.
\end{definition}

There are several well-known properties of the treewidth. First, the treewidth does not decrease under taking minors of a graph.
\begin{lemma}\label{lemma:minor_tw}
Let $H$ be a minor of $G$, i.e., it is a graph obtained by contracting edges of a subgraph of $G$. Then $tw(H) \leq tw(G)$.
\end{lemma}

See Ref.~\cite{robertson1986graph} for a proof.
Next, we have a bound for the treewidth of the line graph as follows~\cite{markov2008simulating}:

\begin{lemma}\label{lemma:line_graph_tw}
    For any graph $G$ with the maximum degree $\Delta(G)$, the treewidth of the line graph is bounded by
    \begin{align}
        (tw(G) -1)/2 \leq tw(L(G)) \leq \Delta(G)(tw(G) + 1) - 1
    \end{align}
\end{lemma}

Under this set-up, Ref.~\cite{markov2008simulating} proved the following result:
\begin{theorem}
    Let us consider a tensor network with whose underlying graph is $G=(V,E)$. Here, each tensor is assigned to a vertex and an edge $(v,w) \in E$ describes that the corresponding indices of tensors $v$ and $w$ are contracted.
    Then one can contract the tensor network in time $O(|V|^{O(1)} \exp[tw(L(G))])$.
\end{theorem}

\subsection{Tensor network algorithm for sampling from an IQP circuit over a graph}
We now use above lemmas to provide a tensor network algorithm for sampling from an instance of the IQP circuit.
In the main text, we consider an IQP circuit whose interaction graph is given by $G=(V,E)$. The circuit associated with this graph has $N=|V|$ qubits, \textsf{RZ} gate for each qubit $e^{-i\theta_i Z_i}$, and \textsf{IsingZZ} gate between $i$ and $j$-th qubits ($e^{-i Z_i Z_j \phi_{ij}}$) if $(i,j)\in E$. 
Thus, for a given output binary string $x_1 \cdots x_N$, we have
\begin{align}
    \psi_{\theta_i,\phi_{ij}}(x_1 \cdots x_N) = \braket{x_1 \cdots x_N | F_{\theta_i, \phi_{ij}} |+}^{\otimes N}, \qquad
    p_{\theta_i,\phi_{ij}}(x_1\cdots x_N) = | \psi_{\theta_i,\phi_{ij}}(x)|^2.
\end{align}
As we study a tensor network algorithm for given fixed parameters $\{\theta_i\}$ and $\{\phi_{ij}\}$, let us omit those parameters in further description for clarity.

We now consider a tensor network description of the circuit.
Following the standard description~\cite{markov2008simulating}, we convert the input state for each qubit $\ket{+}$ and the state we project to $\ket{x_i}$ to rank-1 tensors. Next, each \textsf{RZ} gate becomes a rank-2 tensor with one input and one output. Finally, the \textsf{IsingZZ} gate is converted into rank-4 tensor with two inputs and two outputs.

\begin{figure}
    \centering
    \includegraphics[width=0.5\textwidth]{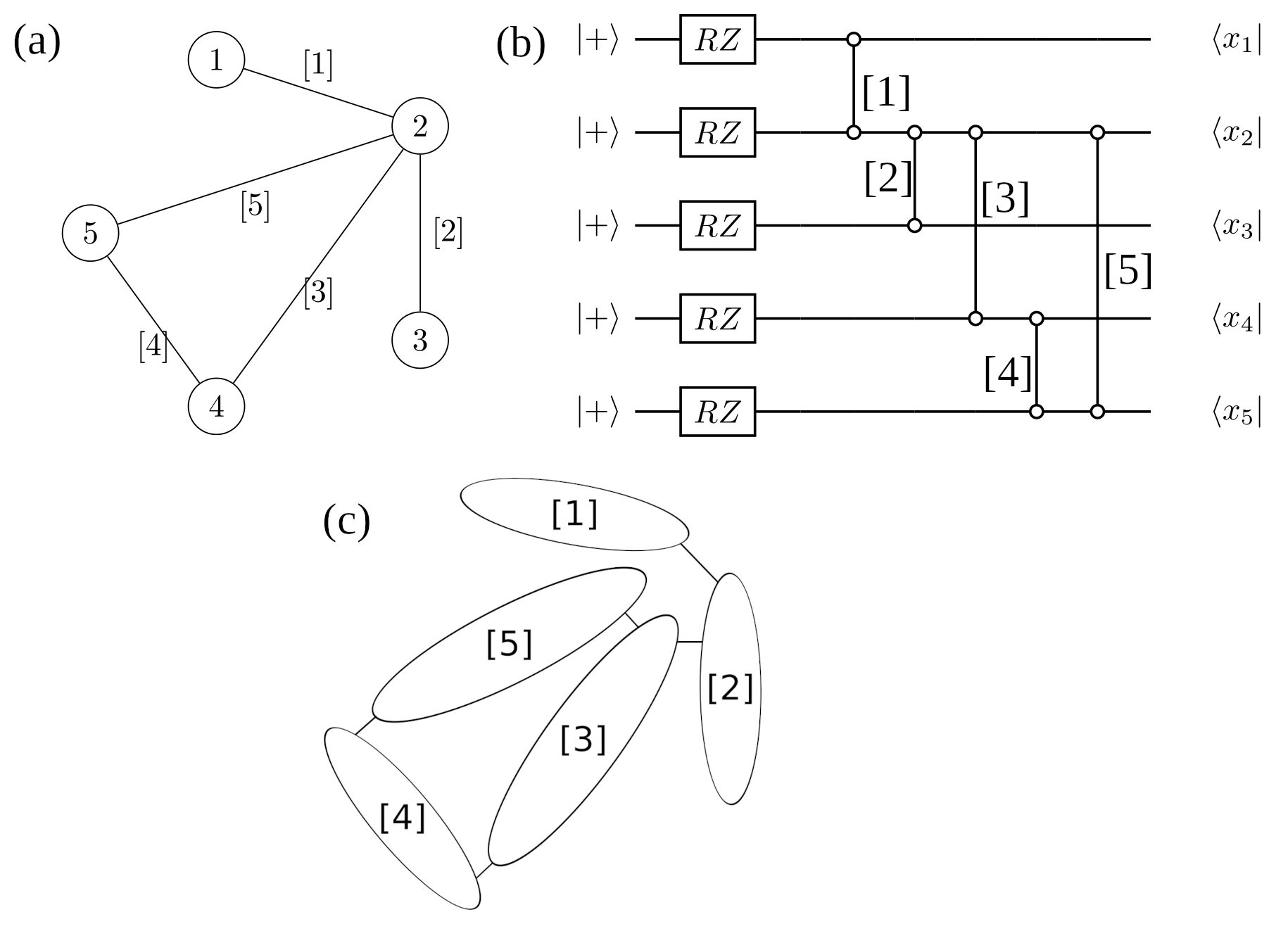}
    \caption{Converting an IQP circuit whose two-qubit gates applies to edges of a graph to a tensor network. (a) A graph instance $G$ and (b) the associated circuit. We use two-qubit gates with open ends to indicate the \textsf{IsingZZ} gate. (c) By contracting the single qubit gates, input states, and projected output states, we obtain a tensor network whose underlying graph is a subgraph of the line graph of $G$.}
    \label{fig:graph_to_tns}
\end{figure}

From this tensor network representation, we can readily contract single qubit gates as well as input and output states (i.e., they are merged into two-qubit gates).
The tensor network after this contraction becomes a graph whose nodes are the tensors for the two-qubit gates (contracted with input/output states) and edges are between nearby two-qubit gates. We show in Fig.~\ref{fig:graph_to_tns} how a original graph to construct a circuit is converted to a tensor network.
A caveat is that the exact tensor network description of the circuit depends on the ordering of two-qubit gates even though the result does not.

A tensor in the final expression has zero, two, three, or four legs depending on the position of the two-qubit gate in the circuit. If a two-qubit gate had one wire directly connected to the input or output wire, it has three legs (as one leg is already contracted).
Likewise, if its two wires were connected to input or output wires, and other wires are connected to other two-qubit gates, the resulting tensor has two legs.
For example, let us consider a two-qubit gate whose one output wire is projected to $\bra{x_i}$. Then the tensor is given as $T^{[a],x_i}_{jkl}$ where $a$ is the index of the tensor, $x_i$ is the given bit we project the state to, and $j,k,l$ are internal indices represents wires connected to other tensors, which are summed over when we contract the tensor network.
Using this notation, we can formally express a component of the output state as
\begin{align}
    \psi(x_1 \cdots x_N) = \mathrm{tTr}\{ T^{[a]}, T^{[a'] x_i},  T^{[a''] x_j x_k}\},
\end{align}
where $\mathrm{tTr}$ indicate a tensor network contraction, $\{T^{[a]}\}$ is a set of tensors without output bit indices (the corresponding two-qubit gate is not adjacent to any projection $\bra{x_i}$), $\{T^{[a']x_i}\}$ are tensors with one output bit index (the two-qubut gate adjacent to one projection), and $\{T^{[a'']x_j x_k}\}$ are tensors with two output bit indices (the gate adjacent to two projections).
In the expression, we omitted internal indices as they are summed over.

Let us denote $\tilde{G}$ by a graph for this tensor network description. A simple observation is that $\tilde{G}$ is a subgrapgh of the line graph $L(G)$ where the maximum degree of vertices is less than $4$.
This is because we assign a two-qubit gate for each edge of $G$, which becomes a node in $\tilde{G}$, and an edge between two nodes in $\tilde{G}$ exists only if they share the same wire (and if they are also adjacent).
By applying lemmas~\ref{lemma:minor_tw} and \ref{lemma:line_graph_tw}, we know that the treewidth of $\tilde{G}$ is upper bounded by
\begin{align}
    tw(\tilde{G}) \leq \Delta(G)(tw(G) + 1) -1.
\end{align}
In addition, the treewidth of the line graph of $\tilde{G}$ is bounded by 
\begin{align}
    tw(L(\tilde{G})) \leq \Delta(\tilde{G}) (tw(\tilde{G}) + 1) - 1 \leq 4 tw(\tilde{G}) + 3 \leq 4 \Delta(G) (tw(G) + 1) - 1
\end{align}
where we have used $\Delta(\tilde{G}) \leq 4$ (as each tensor has at most 4 legs) for the second inequality.
Thus the overall complexity of obtaining $\psi(x_1 \cdots x_N)$ from an IQP circuit over a graph $G=(V,E)$ is $O(|\tilde{G}|^{O(1)} \exp[tw(L(\tilde{G}))])=O(|E|^{O(1)} \exp[tw(4 \Delta(G) tw(G))])$.

We next show that the output probability distribution can also be computed with the same complexity.
This is because we can write the output probability as
\begin{align}
    p(x_1 \cdots x_N) = \mathrm{tTr}\{ T^{[a]} \otimes \overline{T}^{[a]}, T^{[a'] x_i} \otimes \overline{T}^{[a'] x_i},  T^{[a''] x_j x_k} \otimes \overline{T}^{[a''] x_j x_k}\},
\end{align}
where $\overline{T}$ is the complex conjugate of the tensor. Thus one can compute the output probability using a tensor network over the same graph but with doubled internal dimension.
Likewise, one can also compute the marginal distribution $p(\pmb{x}_S)$ for a subset $S \subset \{1, \cdots, N\}$ by using tensors summed over local indices $x_i$ if $i \in [n] \setminus S$ where $[n]=\{1,\cdots,N\}$.
Therefore, we can compute all marginal probabilities $p(\pmb{x}_S)$ for a given $\pmb{x}_S$ in $O(|E|^{O(1)} \exp[tw(4 \Delta(G) tw(G))])$.

We then consider a sampling algorithm for this circuit. We first decompose the probability to nested conditionals
\begin{align}
    p(x_1\cdots x_N ) = p(x_1) \prod_{i=2}^N p(x_i|x_1 \cdots x_{i-1}).
\end{align}
Then we can sample from the distribution using the following procedure. First, we compute $p(x_1=0)$ and choose a random number $r$ from the uniform distribution $\mathcal{U}_[0,1]$. Then choose $x_1 = 0$ if $r < p(x_1=0)$ or $0$ otherwise.
We next compute the conditional probability $p(x_2|x_1) = p(x_1 x_2)/p(x_1)$. As we already know $p(x_1=0)$ and $p(x_1=1)=1-p(x_1=0)$, we only need to compute $p(x_2 x_1)$ for a given $x_1$.
We can sample $x_2$ from the distribution using the same steps.
By repeating these steps, we can sample one bitstring $x_1\cdots x_N$. The overall time complexity of this algorithm comes from computing the marginal distributions. As we need $N$ marginal distributions, the algorithm takes $O(N |E|^{O(1)} \exp[tw(4 \Delta(G) tw(G))])$ to generate a single bitstring.

Finally, we consider the case when a graph $G$ is an instance of the Erd\H{o}s-R\'{e}nyi graph, where each edge is chosen with probability $q$.
When $q < 1/N$, it is proved that the treewidth of the graph is almost surely a constant~\cite{lee2012rank}.
In addition, it is shown in Ref.~\cite{bremner2017achieving} that there is $\gamma > 0$ such that $\mathrm{Pr}[\Delta(G) \geq 2 \gamma \log (N)] \leq N^{1-\gamma/4}$ when $q \leq \gamma \log(N) / N$. As this condition is satisfied for $\gamma =8 $ when $q < 1/N$ and $N \geq e^{1/8} \approx 1.13$, $1-O(1/N)$ fraction of the Erd\H{o}s-R\'{e}nyi graph instances has $\Delta(G) \leq 16 \log N$.
Thus our sampling algorithm takes $\mathcal{O}(N |E|^{O(1)} \exp[tw(4 \Delta(G) tw(G))]) = \mathcal{O}(N \cdot N^2 \exp[ const \cdot log(N)]) = \mathcal{O}(\mathrm{poly}(N))$ time.

\section{Anti-concentration of the output distribution of the IQP circuit over the Erd\H{o}s-R\'{e}nyi graph}\label{app:anticoncentration}
Ref.~\cite{bremner2017achieving} considered IQP-type circuits over the Erd\H{o}s-R\'{e}nyi graph where the single and two-qubit diagonal gates are given by \textsf{RZ} and \textsf{CPhase} gates with discrete random angles, respectively.
In this case, the authors proved that there is $\gamma > 0$ such that the output distribution of the circuit is anticoncentrated when the connection probability of the random graph is larger than $\mathcal{O}( \log N /N)$

In this section, we provide a proof that the output distribution is anticoncentrated under the same probability condition even when our gate set is given by $\{e^{-i \theta_i Z_i}\}\cup \{e^{-i Z_i Z_j \phi_{i,j}}\}$ (\textsf{RZ} and \textsf{IsingZZ} gates) with continuous random angles $\theta_i$ and $\phi_{i,j}$, which are sampled from the uniform distribution over $[0, 2\pi]$.
We prove the anticoncentration property by showing that $Z = \mathbb{E}_U[p_U(x)^2]$ is smaller than $\alpha/2^{2N}$ for a constant $\alpha>1$, where the average is over all possible parameters of the circuit with the given connectivity, and $p_U(x) = |\braket{x|U|0}|^2$ is the probability to obtain $x$ for the given circuit. Then Paley–Zygmund inequality implies that, for any $0 \leq \beta \leq 1$
\begin{align}
    \mathrm{Pr}_{U} \Bigl[ p_U(x) \geq \beta 2^{-N} \Bigr] \geq \alpha^{-1} (1-\beta)^2,
\end{align}
which implies the anticoncentration of the output distribution.

Formally, we write our result as follows:
\begin{proposition}
    Consider an IQP circuit defined over the random graph with the connection probability $q$ . Then there is a constant $\gamma>0$ and $N_0\in \mathbb{N}$ such that the output probability of the circuit is anticoncentrated if $q=(6/\gamma) \log N /N$ for all $N \geq N_0$.
\end{proposition}

We prove this by obtaining the upper bound of $Z$. First, $Z=\mathbb{E}_U[p(x)^2]=\mathbb{E}_U[p(0)^2]$ from the symmetry of the circuit. We first have
\begin{align}
    |\braket{0|U|0}|^4 &= 2^{-4N} \sum_{x,y,z,w \in \{-1,1\}^N} \exp \Bigl[i \sum_{i < j} \phi_{ij} (x_i x_j + y_i y_j - z_i z_j - w_i w_j)\Bigr] \exp \Bigl[ -i \sum_i \theta_i (x_i+y_i-z_i-w_i) \Bigr] \nonumber \\
    &= 2^{-4N} \sum_{x,y,z,w \in \{-1,1\}^N} \exp \Bigl[i \sum_{i < j} \phi_{ij} (x_i x_j + y_i y_j - z_i z_j - w_i w_j)\Bigr] \exp \Bigl[ -i \sum_i \theta_i (x_i+y_i+z_i+w_i) \Bigr].
\end{align}
Averaging this expression over $\{\theta_i\}$ only remains terms with $c_i(x,y,z,w):=x_i+y_i+z_i+w_i = 0$ in the summation.

Given that $\phi_{ij}=0$ with the probability $1-q$ and have a random value with the probability $q$, we also have
\begin{align}
    &\mathbb{E} \Bigl\{ \exp \bigl[i \sum_{i < j} \phi_{ij} (x_i x_j + y_i y_j + z_i z_j + w_i w_j)\bigr] \Bigr\}= (1-q) 1 + q \delta_{x_i x_j + y_i y_j - z_i z_j - w_i w_j}  \nonumber \\
    &\quad = \begin{cases}
        1 &\text{if } x_i x_j + y_i y_j - z_i z_j - w_i w_j = 0\\
        1-q &\text{otherwise}
    \end{cases}.
\end{align}
Thus, for $F_{ij}(x,y,z,w):=x_i x_j + y_i y_j - z_i z_j - w_i w_j$, we have
\begin{align}
    Z = 2^{-4N} \sum_{c_i(x,y,z,w)=0} (1-q)^{|\{(i < j) : F_{i,j} \neq 0\}|}, \label{eq:prob_sqaure_expectation}
\end{align}
where the summation is over all $x,y,z,w \in \{-1,1\}^N$ which satisfies $c_i(x,y,z,w)=0$ for all $i$.

We now introduce a variable $k=\{k_1,\cdots, k_N\}$ where each $k_i=\{x_i,y_i,z_i,w_i\}$ is a 4-dimensional vector $\in \{-1,1\}^4$. We note that, among 16 possible values of $k_i$, only $6$ values satisfy the constraints $c_i$. The following table shows all possible values of $k_i$ with the constraints $c_i$.

\begin{center}
\begin{tabular}{|C{2em}||C{2em} C{2em} C{2em} C{2em} C{2em} C{2em}|} 
 \hline
 & $a$ & $b$ & $c$ & $d$ & $e$ & $f$ \\ 
 \hline\hline
 $x_i$ & 1 & 1 & 1 & -1 & -1 & -1 \\ 
 \hline
 $y_i$ & 1 & -1 & -1 & 1 & 1 & -1 \\
 \hline
 $z_i$ & -1 & 1 & -1 & 1 & -1 & 1 \\
 \hline
 $w_i$ & -1 & -1 & 1 & -1 & 1 & 1 \\
 \hline
\end{tabular}
\end{center}

Here, we used $a$-$f$ to indicate each possible value of $\{x_i,y_i,z_i,w_i\}$. 
Thus the summation in Eq.~\eqref{eq:prob_sqaure_expectation} can be rewritten as the summation over $\{a,b,c,e,d,f\}^N$.
In addition, we also see that all possible pairs of $(k_i, k_j)$ with $F_{i,j}\neq 0$ are $\{(b,c), (b,d), (c,b), (c,e), (d,b), (d,e), (e,c), (e,d)\} =: N$.
As $d=-c$, $e=-b$, $f=-a$, we observe that $(k_i, k_j) \in N$ if and only if $(k_i, -k_j) \in N$.
Using this notation, we write $Z$ as 
\begin{align}
    Z = 2^{-4N} \sum_{k \in \{a,b,c,-a,-b,-c\}^N} (1-q)^{|\{(k_i, k_j) \in N \}|}
\end{align}
As flipping the sign of $k_i$ for any $i$ does not change the summand, we can simplify the expression to
\begin{align}
    & = 2^{-3N} \sum_{k \in \{a,b,c\}^N} (1-q)^{|\{(k_i, k_j) \in N\}|}.
\end{align}

Let $A = |\{i: k_i = a\}|$, $B = |\{i: k_i = b\}|$, $C = |\{i: k_i = c\}|$ be the number of occurrences of $a$, $b$, $c$ in $k$. Then we obtain
\begin{align}
    {|\{(i<j): (k_i, k_j) \in N \}|} = BC
\end{align}
and
\begin{align}
    Z =& 2^{-3N} \sum_{A,B,C \geq 0, A+B+C=N} M(A,B,C) (1-q)^{BC},
\end{align}
where 
\begin{align}
    M(A,B,C) = {N \choose A} {N-A \choose B} {N-A-B \choose C} = \frac{N!}{A!B!C!}
\end{align}
is the number of $k$s with $(A,B,C)$.
We compute the upper bound of the summation by splitting the cases.
First, we consider the case $A \leq (1-\gamma) N$ for some $0 < \gamma < 1/2$ satisfying ${N \choose \gamma N} \leq 2^{N/2}/N$. For example, $\gamma=1/10$ satisfies this condition (which is from ${n \choose k} \leq 2^{n H(k/n)}$ where $H(p)$ is the Shannon entropy in base 2).
Then at least one of $B$ or $C$ is larger than $\gamma N/2$ (w.l.o.g. assume $C \geq \gamma N/2$).

The summation for this case is upper bounded by
\begin{align}
    \sum_{A=0}^{(1-\gamma)N} \sum_{B=0}^{N-A} {N \choose A} {N-A \choose B} (1-q)^{BC} &\leq 2 \sum_{A=0}^{(1-\gamma)N} \sum_{B=0}^{N - A} {N \choose A} {N-A \choose B} (1-q)^{\gamma N B/2} \nonumber \\
    &\leq 2 \sum_{A=0}^{(1-\gamma)N} {N \choose A} \sum_{B=0}^{(1-\gamma)N - A} N^B (1-q)^{\gamma N B/2} \nonumber \\
    &\leq 2 \times 2^N \sum_{B=0}^{N} N^B e^{-q\gamma N B/2}.
\end{align}
where we used ${N-A \choose B} \leq (N-A)^{B}/B! \leq N^B$ for the second inequality and $1-q \leq e^{-q}$ for the last inequality.
When $q \geq (6/\gamma) \log N /N$,
\begin{align}
    N^B e^{-q\gamma N B / 2} \leq N^B e^{-3 B \log N} \leq N^{-2B}.
\end{align}
Thus $\sum_{B=0}^{N} N^B e^{-q\gamma N B/2} \leq 1 + 1/N \leq 2$, and this part contribute at most $4 \times 2^{-2N}$ to $Z$.

Next, we consider $A \geq (1-\gamma) N$. Then we obtain
\begin{align}
     \sum_{A=(1-\gamma)N}^{N} \sum_{B=0}^{N-A} {N \choose A} {N-A \choose B} (1-q)^{BC} &\leq \sum_{A=(1-\gamma)N}^{N} \sum_{B=0}^{\gamma N} {N \choose A} {N \choose B} \nonumber \\
     &\leq \sum_{A=0}^{\gamma N} {N \choose A} \sum_{B=0}^{\gamma N} {N \choose B} \nonumber \\
     &\leq \Bigl[ \frac{2^{N/2} \gamma N}{N} \Bigr ]^2 \leq 2^N.
\end{align}

To summarize, we showed that $Z \leq 5 \times 2^{-2N}$ if $q \geq (6/\gamma) \log N / N$ for some constant $\gamma$.
This completes the proof that the distribution anti-concentrates.

In addition, we also can obtain a lower bound of $Z$ as
\begin{align}
    Z = 2^{-3N} \sum_{A=0}^{N} \sum_{B=0}^{N-A} {N \choose A} {N-A \choose B} (1-q)^{BC} &\geq 2^{-3N} \sum_{A=0}^{N} \sum_{B=0}^{N-A} {N \choose A} {N-A \choose B} (1-q)^{NB} \nonumber \\
    &= 2^{-3N} (2 + (1-q)^N)^N,
\end{align}
where we used $C \leq N$ for the first inequality and the property of the multivariate coefficients for the last equality.
If $q \leq \lambda /N$ for some $\lambda > 0$, we have
\begin{align}
    (1-q)^N \geq (1-\lambda /N)^N =[(1-\lambda /N)^{N/\lambda}]^\lambda.
\end{align}
Note that $(1-\lambda /N)^{N/\lambda}$ approaches to $e^{-1}$ as $N$ increases, and we can find $N_0 := N_0(\lambda)$ such that 
\begin{align}
    \bigl| (1-\lambda /N)^{N/\lambda} - e^{-1} \bigr| \leq 1/12,
\end{align}
for all $N \geq N_0$. Then we have the lower bound
\begin{align}
    Z \geq 2^{-3N} \biggl[ 2 + \Bigl( \frac{11e^{-1}}{12} \Bigr)^\lambda \biggr]^N.
\end{align}
Following Ref.~\cite{dalzell2022random}, we can conclude that the output distribution is not anticoncentrated when $q < \lambda/N$ for any constant $\lambda>0$.

\begin{figure}
    \raisebox{-0.5\height}{\includegraphics[width=0.35\textwidth]{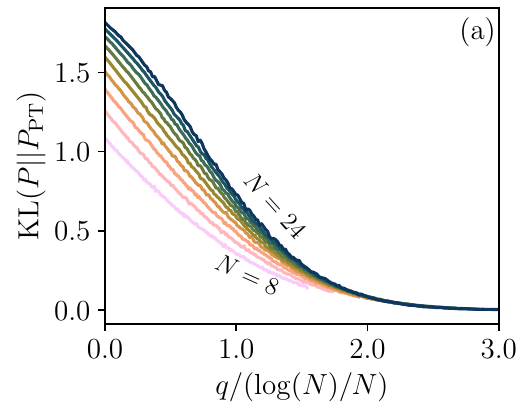}}
    \raisebox{-0.5\height}{\includegraphics[width=0.5\textwidth]{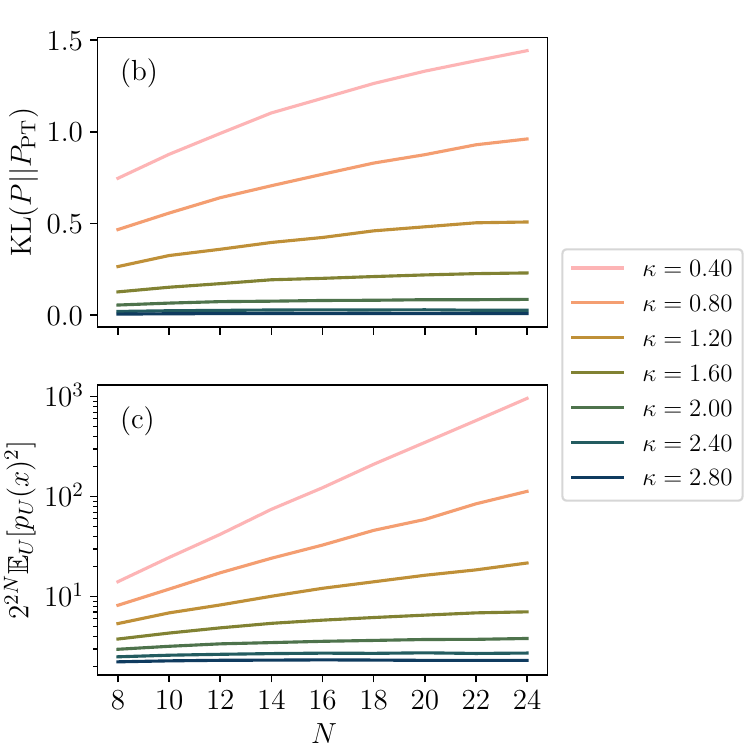}}
    \caption{(a) KL divergence between the output distribution of the IQP circuits over the random graph with the connection probability $q$ and the Porter-Thomas distribution as a function of $q/(\log(N)/N)$.
    For different values of $q=\alpha \log N /N$, we plot (b) the KL divergence between the output distribution and the Porter-Thomas distribution as well as (c) the squared probability $p(x)^2$ averaged over all circuit instances and the output strings $x$. These figures show that the output distribution at $\alpha \approx 2.0$ is anticoncentrated but does not follow the Porter-Thomas distribution.}
    \label{fig:anticon_vs_porter_thomas}
\end{figure}

\section{Anticoncentration versus Porter-Thomas distribution for the IQP circuit} \label{app:anticon_vs_pt}
One of our main results in the paper is that there exists a parameter regime that the output distribution is anticoncentrated albeit it does not follow the Porter-Thomas distribution.
In this section, we provide additional numerical evidence supporting our argument.
Precisely, we find a parameter regime such that the KL divergence between the output distribution of the IQP circuits and the Porter-Thomas distribution converges to a constant (as $N$ increase), but $Z = \mathbb{E}_U[p_U(x)^2] \leq \alpha / 2^{2N}$ for some constant $\alpha > 0$ which is a sufficient condition for the anticoncentration (see Sec.~\ref{app:anticoncentration}).

We show a plot of the KL divergence as a function of $q/(\log N / N)$ in Fig.~\ref{fig:anticon_vs_porter_thomas}(a).
It suggests that the KL divergence does not decreases with $N$, but diverges or converges to a constant for any fixed $q/(\log N / N)$, which implies that the distribution never becomes the Porter-Thomas distribution for a fixed $q/(\log N / N)$.
We further compare the scaling of the KL divergence and $Z=\mathbb{E}_U[p(x)^2]$ that we studied in Sec.~\ref{app:anticoncentration} in Fig.~\ref{fig:anticon_vs_porter_thomas}(b) and Fig.~\ref{fig:anticon_vs_porter_thomas}(c).
Plots are from different values of $\kappa \in [0.4, \cdots 2.8]$ where $q=\kappa \log N/N$.
We see that the KL divergence converges to a constant as $N$ increases when $\kappa \gtrsim 1.6$. This implies that the output distribution maintains a distance to the Porter-Thomas distribution even for large $N$.
In contrast, we see that $Z$ scales with $\beta(\kappa) / 2^{2N}$, which is the condition for the anticoncentration (Sec.~\ref{app:anticoncentration}).
This is the parameter regime that the output distribution is anticoncentrated but does not follow the Porter-Thomas distribution.

\section{Entanglement properties of the IQP circuit} \label{app:ent_iqp}
In this section, we diagnose the phase transition of IQP circuits over the random graph using entanglement spectra.
As the original circuit has no entanglement at $q=0$, we consider a slightly modified circuit with $N-1$ body interactions instead of the conventional one body terms ($\{Z_i\}$).
In other words, we consider a diagonal unitary operator $D_{\theta_i, \phi_{ij}}$ such that 
\begin{align}
D_{\theta_i, \phi_{ij}}=\exp[i \sum_i \theta_i Z_{\neg i}] \exp[i \sum_{i < j} \phi_{ij} Z_i Z_j].
\end{align}
where $Z_{\neg i} = \prod_{j \neq i} Z_j$ is the product of Pauli-Z operators of all sites except $i$.
Thus we consider the IQP circuit given by
\begin{align}
    \psi(x) = \braket{x|H^{\otimes N} D_{\theta_i, \phi_{ij}}|+}^{\otimes N} \label{eq:circuit_D_output_st}
\end{align}
and the output distribution from the circuit
\begin{align}
r(x) = \bigl| \braket{x|H^{\otimes N} D_{\theta_i, \phi_{ij}}|+}^{\otimes N} \bigr|^2.
\end{align}
This choice ensures that the output states are always entangled, which makes the classification of phases via the entanglement spectrum~\cite{pal2010many} meaningful. 

Even though this setup changes the entanglement properties of the output states, it does not change the complexity as there is a simple relation between the distribution $r(x)$ and $p(x)$ that from the IQP circuit with \textsf{RZ} gate.
Precisely, $p(x) = r(x)$ if $x$ has even parity (i.e. $\sum_{k=1}^N x_k$ is even) and $p(x) = r(\overline{x})$ otherwise, with $\overline{x} = (1-x_i)$, which we prove in the following subsection.

\subsection{Output distribution of the IQP circuit with $Z_{\neg i}$}

In this subsection, we show that using $Z_{\neg i}$ instead of $Z_i$ for the IQP circuits does not change the complexity of sampling, which follows from the theorem below.
\begin{theorem}
We consider two diagonal unitary operators $D_{\theta_i, \phi_{ij}}$ and $F_{\theta_i, \phi_{ij}}$, where $D$ is composed of $Z_{\neg i} = \prod_{k \neq i} Z_k$ and two-qubit phase shift gates $\{e^{\theta_i Z_{\neg i}}, e^{\phi_{ij} Z_i Z_j}\}$ whereas $F$ is constructed with one- and two-qubit gates $\{e^{\theta_i Z_i}, e^{\phi_{ij} Z_i Z_j}\}$. We write $r(x) = |\braket{x| H^{\otimes N}D_{\theta_i, \phi_{i,j}}|+}^{\otimes N}|^2$ and $p(x) = |\braket{x|H^{\otimes N}F_{\theta_i, \phi_{ij}}|+}^{\otimes N}|^2$.
When the total number of qubits is even, we have
\begin{align}
    p(x) = \begin{cases}
    r(x) , \text{if }\sum_k x_k\text{ is even} \\
    r(\overline{x}) , \text{otherwise}
    \end{cases}
\end{align}
where $\overline{x}_i = (1-x_i)$. Thus a sampler for $p(x)$ can also sample from $r(x)$ (and vice versa).
\end{theorem}

\begin{proof}
Direct calculation gives
\begin{align}
    &\braket{x| H^{\otimes N} D_{\theta_i, \phi_{i,j}}|+}^{\otimes N}=\braket{x| H^{\otimes N} \exp \Bigl[i\sum_{i < j} \phi_{ij} Z_i Z_j +\sum_{i} \theta_i Z_{\neg i}\Bigr]| +}^{\otimes N} \\
    &\qquad = \frac{1}{\sqrt{2^N}}\sum_y \braket{x | H^{\otimes N} \exp \Bigl[i\sum_{i < j} \phi_{ij} Z_i Z_j +i \sum_{i} \theta_i Z_{\neg i}\Bigr]|y}\\
    &\qquad=\frac{1}{\sqrt{2^N}}\sum_y \braket{x | H^{\otimes N} | y} \exp \Bigl[i \sum_{i < j} \phi_{ij} (1-2y_i)(1-2y_j) + i \sum_i \theta_i (-1)^{\sum_k y_k}(1-2y_i)\Bigr]
\end{align}
where we used $Z_i\ket{y} = (-1)^{y_i} = (1-2y_i)$ and $Z_{\neg i}\ket{y} = \prod_{k \neq i} Z_k \ket{y} = Z^{\otimes N} Z_i \ket{y} = (-1)^{\sum_k y_k} (-1)^{y_i}$.
By decomposing the summation over $y \in \{0, 2^N-1\}$ into $S_{\rm even} = \{y:\sum_k y_k \text{ is even}\}$ and $S_{\rm odd} = \{y:\sum_k y_k \text{ is odd}\}$, we have

\begin{figure}
    \centering
    \includegraphics[width=0.95\linewidth]{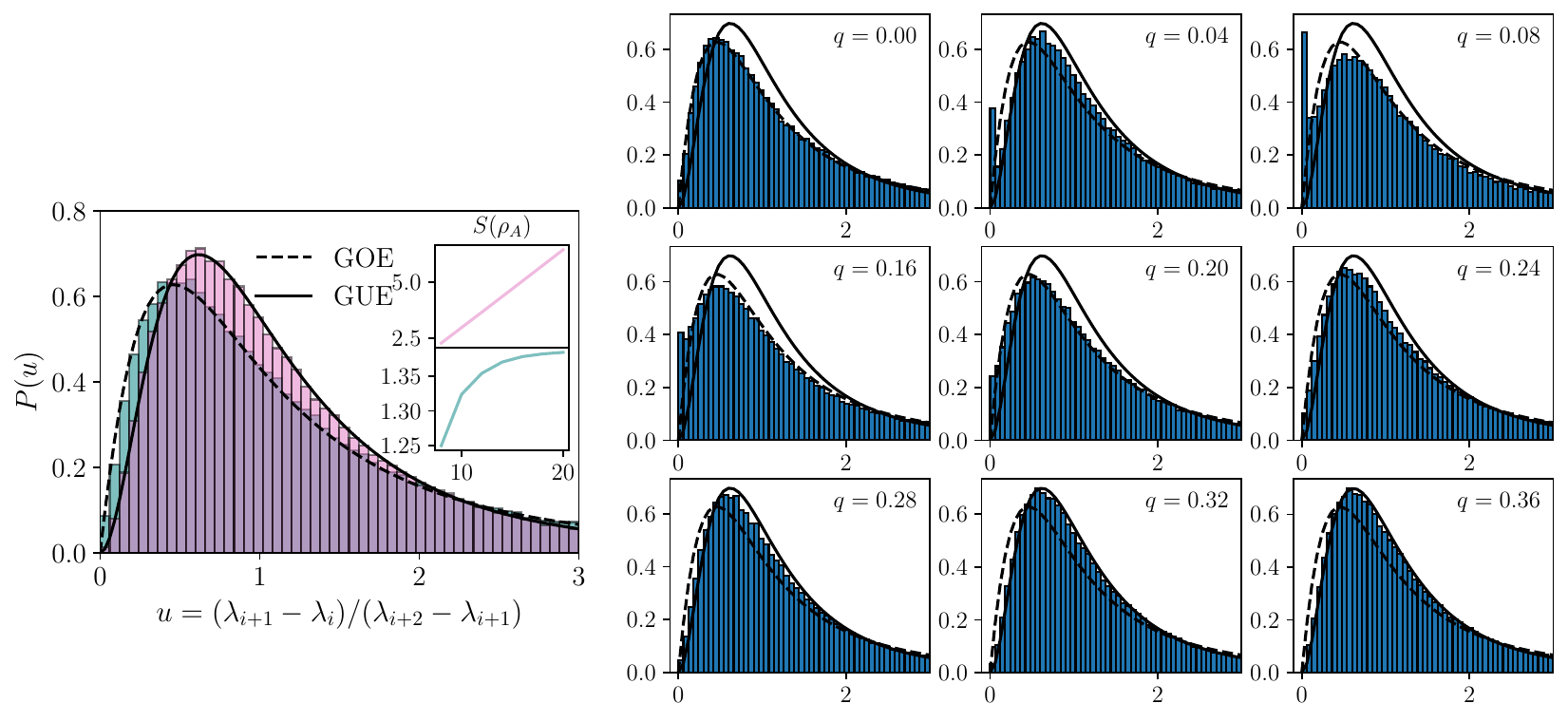}
    \caption{Entanglement properties of the IQP circuit over the random graph. Left: Distribution of the ratio between adjacent entanglement spectra for $q=0.0$ (pink) and $q=1$ (teal). We see that the distribution follows that from the Gaussian orthogonal ensemble (GOE) when $q=0.0$ whereas it shows the Gaussian unitary ensemble (GUE) property when $q=1.0$. Inset shows the entanglement entropy as a function of the size of the system. Right: Entanglement spectra gaps for different values of $q$. As $q$ increases, the distribution changes from the GOE to GUE.}
    \label{fig:iqp_entanglement}
\end{figure}

\begin{align}
&= \frac{1}{\sqrt{2^N}} \biggl\{ \sum_{y \in S_{\rm even}} \braket{x | H^{\otimes N} | y} \exp \Bigl[i \sum_{i < j} \phi_{ij} (1-2y_i)(1-2y_j) + i\sum_i \theta_i (1-2y_i)\Bigr] \nonumber \\
&\qquad \qquad + \sum_{y \in S_{\rm odd}}\braket{x | H^{\otimes N} | y} \exp \Bigl[i \sum_{i < j} \phi_{ij} (1-2y_i)(1-2y_j) - i\sum_i \theta_i (1-2y_i)\Bigr] \biggr\} \\
&= \frac{1}{\sqrt{2^N}} \biggl\{ \sum_{y \in S_{\rm even}} \braket{x | H^{\otimes N} | y} \exp \Bigl[i \sum_{i < j} \phi_{ij} (1-2y_i)(1-2y_j) + i\sum_i \theta_i (1-2y_i) \Bigr] \nonumber \\
&\qquad \qquad + \sum_{y' \in S_{\rm odd}}\braket{x | H^{\otimes N} | \overline{y'}} \exp \Bigl[i \sum_{i < j} \phi_{ij} (1-2y_i')(1-2y_j') + i\sum_i \theta_i (1-2y_i')\Bigr] \biggr\}\\
&= \frac{1}{\sqrt{2^N}} \Bigl\{ \sum_{y \in S_{\rm even}} \braket{x | H^{\otimes N} | y} \braket{y|F_{\theta_i, \phi_{i,j}}|y} + \sum_{y' \in S_{\rm odd}} \braket{x | H^{\otimes N} | \overline{y'}} \braket{y'|F_{\theta_i, \phi_{i,j}} | y'} \Bigr\}
\end{align}
where we changed the variable $y_i= 1-y_i'$ in the latter sum ($y'$ has the same parity as $y$ as we only consider even $N$) and used the notation $\ket{\overline{y}} = \otimes_i\ket{1-y_i}$.
As $\braket{x|H^{\otimes}|y} = (-1)^{\sum_k x_k y_k}/\sqrt{2^N}$ and $\braket{x|H^{\otimes}|\overline{y}} = (-1)^{\sum_k x_k (1-y_k)}/\sqrt{2^N} = (-1)^{\sum_k x_k} (-1)^{\sum_k x_k y_k}/\sqrt{2^N}$, we see that $p(x)=r(x)$ when $\sum_k x_k$ is even.
On the other hand, when $\sum_k x_k$ is odd, we have 
\begin{align}
&= \frac{1}{\sqrt{2^N}} \Bigl\{ \sum_{y \in S_{\rm even}} \braket{x | H^{\otimes N} | y} \braket{y|F_{\theta_i, \phi_{i,j}}|y} + \sum_{y' \in S_{\rm odd}} (-1)\braket{x | H^{\otimes N} | y'} \braket{y'|F_{\theta_i, \phi_{i,j}} | y'} \Bigr\} \\
&= \frac{1}{\sqrt{2^N}} \Bigl\{ \sum_{y \in S_{\rm even}} \braket{\overline{x} | H^{\otimes N} | y} \braket{y|F_{\theta_i, \phi_{i,j}}|y} + \sum_{y' \in S_{\rm odd}} \braket{\overline{x} | H^{\otimes N} | y'} \braket{y'|F_{\theta_i, \phi_{i,j}} | y'} \Bigr\} \\
&= \braket{\overline{x} | H^{\otimes N} F_{\theta_i, \phi_{i,j}} | +}^{\otimes N}
\end{align}
where we have used $\braket{\overline{x} | H^{\otimes N} | y} = (-1)^{\sum_k y_k} \braket{x | H^{\otimes N} | y}$ in the second line. Thus, $p(x) = r(\overline{x})$ if $\sum_k x_k$ is odd.
\end{proof}

\subsection{Entanglement spectra with varying density of two-qubit gates}
We now study entanglement properties of the output state [Eq.~\eqref{eq:circuit_D_output_st}] from IQP circuits with $Z_{\neg i}$.
For each circuit realization, we compute the density matrix after tracing out the half of the qubits ($B=[N/2+1, \cdots, N]$) and obtain the density matrix $\rho_A$.
By diagonalizing $\rho_A$, we obtain an ordered entanglement spectrum $\{\lambda_i\}$.
Then the distribution for the ratio of the adjacent eigenvalues $u = (\lambda_{i+1}-\lambda_{i})/(\lambda_{i+2}-\lambda_{i+1})$ is computed.
We use the ratios as they give better signature~\cite{pal2010many}, and the exact distributions for the Gaussian orthogonal ensemble (GOE) and the Gaussian unitary ensemble (GUE) are known~\cite{atas2013distribution}.

In Fig.~\ref{fig:iqp_entanglement} (Left), we plot the distribution of the entanglement spectra obtained for $N=20$. The distribution clearly follows that from the GOE when $q=0.0$ whereas it follows the GUE when $q=1.0$. We also obtained that the bipartite entanglement entropy (inset) follows the volume law ($S(\rho) \propto N$) and the area law ($S(\rho) = \mathcal{O}(1)$) for $q=1.0$ and $q=0.0$, respectively.

We also plot how the distribution changes as $q$ increases in Fig.~\ref{fig:iqp_entanglement} (Right). One sees that the distribution does not change smoothly, but an initial peak appears for small $q$. We believe that this is a signature of another phase between the GOE and GUE phases that is related to anticoncentration.

\section{Learning the output distribution of the IQP circuit} \label{app:learning_output_dist_iqp}
In the main text, we argued that the output distribution of IQP circuits over the random graph is learnable when $q=0$ (i.e., only with single qubit gates).
On the other hand, we stated that, under a plausible assumption, there is an instance with $N^{\Omega(1)}$ edges whose output distribution is difficult to be learnt. In this section, we prove those statements.

\subsection{Efficient learnability of a distribution with a single binary outcome}
Let us first consider learning a distribution $P$ over a single bit $x \in \{0,1\}$ where $P(x=0) =: p_0$ and $P(x=1) = 1-p_0$.
Our task is reconstructing a probability distribution $Q(x)$ such that $\mathrm{TV}(P,Q) \leq \epsilon$ where 
\begin{align}
    \mathrm{TV}(P,Q) = \frac{1}{2}\sum_{x \in \{0,1\}} |P(x) - Q(x)|
\end{align}
is the total variational distance between two probability distributions.

For a given $\mathcal{N}$ samples, we estimate $p_0$ as $q_0 = \mathcal{N}_0/\mathcal{N}$ where $\mathcal{N}_0$ is the number of $0$s in the samples.
Then one can see the number of observed $\mathcal{N}_0$ follows the binomial distribution $X \sim B(\mathcal{N}, p_0)$. Thus we have
\begin{align}
    \mathrm{Pr}(|p_0 - q_0| \geq \epsilon) &= \mathrm{Pr}( |\mathcal{N}p_0 - \mathcal{N}_0| \geq \mathcal{N} \epsilon) =\mathrm{Pr}(|E(X) - X| \geq \mathcal{N}\epsilon) \nonumber  \\
    & \leq \frac{\mathcal{N} p_0(1-p_0)}{(\mathcal{N}\epsilon)^2} \leq \frac{1}{4 \mathcal{N} \epsilon^2 },
\end{align}
where we used $E(X) = \mathcal{N}p_0$ and the Chevshev's inequality with $\mathrm{Var}(X) = \mathcal{N} p_0 (1-p_0) \leq \mathcal{N}/4$. 

As we can write the total variational distance as
\begin{align}
    \mathrm{TV}(P,Q) &= \frac{1}{2} \Bigl[ | P(x=0) - Q(x=0) | + | P(x=1) - Q(x=1) | \Bigr] \nonumber \\
    &= \frac{1}{2} \Bigl[ | p_0 - q_0 | + | (1-p_0) - (1-q_0) | \Bigr] = |p_0 - q_0|,
\end{align}
we have
\begin{align}
    \mathrm{Pr}\bigl[ \mathrm{TV}(P,Q) \leq \epsilon \bigr] = 1 - \mathrm{Pr}(|p_0 - q_0| > \epsilon) \geq 1 - \frac{1}{4 \mathcal{N} \epsilon^2 }.
\end{align}

Therefore, for any given $\epsilon,\delta \in [0,1]$, we can find $Q$ such that $\mathrm{TV}(P,Q) \leq \epsilon$ with probability $1-\delta$ when
\begin{align}
    \mathcal{N} \geq \frac{1}{4 \epsilon^2 \delta}. \label{eq:learn_sample_size}
\end{align}
As our algorithm finds $q_0$, it is an evaluator. In addition, using a given random number from the uniform distribution $r \in \mathcal{U}_{[0,1]}$, we can construct a circuit which outputs $0$ if $r < q_0$ or $1$, otherwise. Thus the distribution is also learnable w.r.t. a generator.

\subsection{Efficient learnability of product distributions}
We now use the result from the previous section to prove the learnability of a product distribution. For a given $N$, we consider a probability distribution given by $p(x) = \prod_{i=1}^N p_i(x_i)$ where $x=\{x_1,\cdots, x_n\} \in \{0,1\}^N$.
As in the previous subsection, we assume that $\mathcal{N}$ samples from the distribution are provided. In addition, let $\mathcal{N}^i_0$ be the number of samples where the $i$-th bit is $0$.
Because $\{p_i(x_i)\}$ are mutually independent, one can see that each $\mathcal{N}^i_0$ follows the binomial distribution $B(\mathcal{N}, p_i(x=0))$.
Then, using the above result, we can reconstruct a $q_i(x_i)$ which is close to $p_i(x_i)$.
Finally, we can construct the product of the reconstructed distributions $Q(x) = \prod_{i=1}^N q_i(x_i)$ which is close to $P(x)$.

More precisely, as
\begin{align}
    \mathrm{TV}(P, Q) = \frac{1}{2} \sum_{x \in \{0,1\}^N} |P(x) - Q(x)| \leq N \max_i \mathrm{TV}(p_i, q_i),
\end{align}
this quantity can be bounded by $\epsilon$ when $\mathrm{TV}(p_i,q_i) \leq \epsilon/N$. Substituting $\epsilon$ with $\epsilon/N$ in Eq.~\eqref{eq:learn_sample_size}, we know that such a distribution $Q$ is learnable with the probability $1-\delta$ if
\begin{align}
    \mathcal{N} \geq \frac{N^2}{4 \epsilon^2 \delta}.
\end{align}

\subsection{Difficulty of learning the output distribution of the IQP circuit with a sparse connectivity}
We now prove that, under a slightly modified version of the learning parities with noise assumption, the output distribution of the IQP circuit with $N^{\Omega(1)}$ edges is not efficiently learnable.
Formally, the learning parities with noise is defined as follows.

\begin{conjecture}[Learning parities with noise]\label{conj:lpn}
    For any bitstring $s \in \{0,1\}^{k}$, we define a boolean function $f_s(x) = s\cdot x$ for all $x \in \{0,1\}^k$. Then for any $\eta \in (0,1/2)$ the ``noisy parity distribution on $k+1$ bits'' is defined as
    \begin{align}
        P_{s,\eta}(x,y) = \begin{cases}
            2^{-k} (1-\eta)&, \text{ if } y = f_s(x) \\
            2^{-k} \eta    &, \text{ otherwise }
        \end{cases}, \label{eq:lpn}
    \end{align}
    where we consider $(x,y) \in \{0,1\}^{k+1}$ as a bitstring with $k+1$ bits. Then a conjecture is that, there is $\eta \in (0, 1/2)$ and $s \in \{0,1\}^k$ such that there is no classical efficient algorithm for learning $P_{s,\eta}(x,y)$ w.r.t. an evaluator (polynomial in $k$).
\end{conjecture}

However, for our problem, we need a slightly stronger assumption.
\begin{conjecture}\label{conj:lpn_modified}
    For some $\eta \in (0,1/2)$ and $s \in \{0,1\}^k$, there is no classically efficient algorithm for learning $P_{s,\eta}(x,y) \otimes \Id^{k}$, where $\Id^{k}$ is the uniform distribution defined over $k$ bits.
\end{conjecture}
This condition is equivalent to that learning $P_{s,\eta}(x,y)$ is difficult for some $s \in \{0,1\}^k$ with $w(s) < k/2$, where $w(s)$ is the Hamming weight of $s$ (the number of $1$ in the bitstring).
We note that it is possible to prove Conjecture~\ref{conj:lpn_modified} using Conjecture~\ref{conj:lpn} if we use the multiplicative error model.

\begin{proposition}[Hardness of the product distribution]
    Assume that a probability distribution $P(x)$ is classically hard w.r.t. an evaluator.
    Then the probability distribution $P \otimes \Id^k$ (for $k=|x|$) is also difficulty w.r.t. an evaluator, but up to a multiplicative error.
    In other words, there is no classically efficient algorithm finding $Q(x,y)$ such that, for any $\epsilon >0$,
    \begin{align}
        (1-\epsilon) 2^{-k}P(x) < Q(x,y) < (1+\epsilon) 2^{-k}P(x).
    \end{align}
\end{proposition}
\begin{proof}
    Let us assume that there is such a classically efficient algorithm. Then we estimate $P(x)$ with $\tilde{P}(x)=2^k Q(x,0)$. 
    Then, one can easily see that $\mathrm{TV}(\tilde{P},P) \leq \epsilon$, which implies that $P$ is efficiently learnable.
\end{proof}

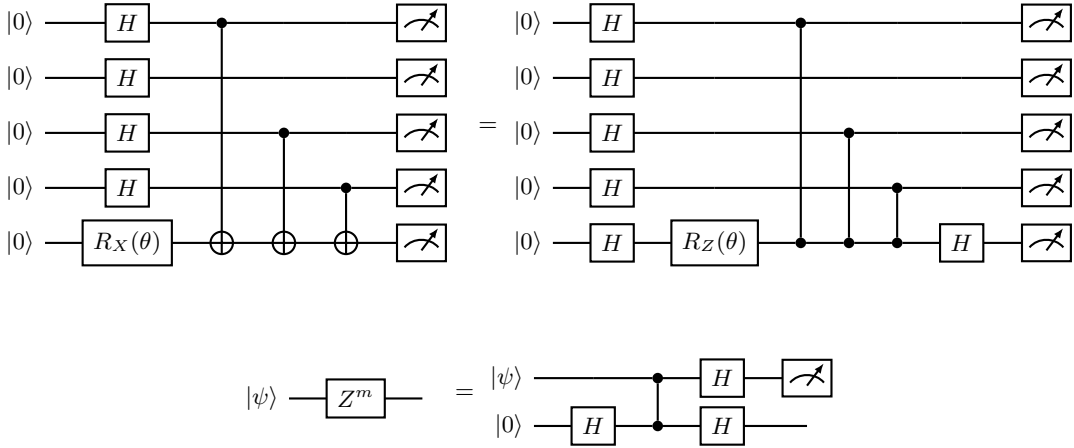
\begin{figure}
    \centering
    \begin{quantikz}[row sep=0.13cm]
    \lstick{$\ket{0}$} & \gate{H}           & \ctrl{4}   & \qw      & \qw      & \meter{} \\
    \lstick{$\ket{0}$} & \gate{H}           & \qw        & \qw      & \qw      & \meter{} \\
    \lstick{$\ket{0}$} & \gate{H}           & \qw        & \ctrl{2} & \qw      & \meter{} \\
    \lstick{$\ket{0}$} & \gate{H}           & \qw        & \qw      & \ctrl{1} & \meter{} \\
    \lstick{$\ket{0}$} & \gate{R_X(\theta)} & \targ{}    & \targ{}  & \targ{}  & \meter{} 
    \end{quantikz}\quad =
    \begin{quantikz}[row sep=0.13cm]
    \lstick{$\ket{0}$} & \gate{H} & \qw                & \ctrl{4}   & \qw        & \qw        & \qw & \meter{} \\
    \lstick{$\ket{0}$} & \gate{H} & \qw                & \qw        & \qw        & \qw        & \qw & \meter{} \\
    \lstick{$\ket{0}$} & \gate{H} & \qw                & \qw        & \ctrl{2}   & \qw        & \qw & \meter{} \\
    \lstick{$\ket{0}$} & \gate{H} & \qw                & \qw        & \qw        & \ctrl{1}   & \qw & \meter{} \\
    \lstick{$\ket{0}$} & \gate{H} & \gate{R_Z(\theta)} & \control{} & \control{} & \control{} & \gate{H} & \meter{} 
    \end{quantikz}\\
    \vspace{1cm}
    \begin{quantikz}[row sep=0.13cm]
    \lstick{$\ket{\psi}$} & \gate{Z^m} & \qw
    \end{quantikz} \quad = 
    \begin{quantikz}[row sep=0.13cm]
    \lstick{$\ket{\psi}$} & \qw      & \ctrl{1}   & \gate{H} & \meter{} \\
    \lstick{$\ket{0}$}    & \gate{H} & \control{} & \gate{H} & \qw
    \end{quantikz} 
    \caption{Circuit identities we used. (Upper) A circuit whose measurement outcomes follow the noise with parity distributions can be transformed into a circuit with Hadamard gates followed by diagonal gates. (Lower) The teleportation circuit we use to add the last Hadamard layer to match with the IQP circuit. Here, $m$ is the measurement outcome from the right circuit.}
    \label{fig:lpn_to_iqp}
\end{figure}

We now consider $D_{s,\eta}(n,k)$, which is defined for $n$-bits whose first $k+1$ follows $P_{s,\eta}$ and the last $n-(k+1)$ bits is a trivial distribution.
In other words, for the trivial distribution $T_l$ over $l$-bits defined as $T_l(x) = 1$ iff $x=0^l$, we have $D_{s,\eta}(n,k) = P_{s,\eta} \otimes T_l$.
The above conjecture says that learning the distribution $D_{s,\eta}(n,k)$ with $k=n^{\Omega(1)}$ is also classically hard (as $\omega(\mathrm{ploy}(k))=\omega(\mathrm{poly}(n))$; see Ref.~\cite{hinsche2022single} for detailed discussion).

Then let us consider a circuit for $N$ qubits initialized with $\ket{0}^{\otimes N}$.
We first apply the Hadamard gate to the first $k$ qubits and the \textsf{RX} gate ($e^{-i X \theta /2}$) to the $k+1$-th qubit.
Then for a given $s = \{s_1,\cdots,s_k\}$, we apply the \textsf{CNOT} gates between $i$-th and the $k+1$-th qubits for all $s_i = 1$. This gives a quantum state
\begin{align}
    \ket{\psi} = \frac{1}{\sqrt{2^k}} \sum_{x=0}^{2^k} \Bigl[ \cos(\theta/2)  \ket{x} \otimes \ket{s \cdot x} \otimes \ket{0^{n-(k+1)}} - i \sin(\theta/2)  \ket{x} \otimes \ket{s \cdot x \oplus 1} \otimes \ket{0^{n-(k+1)}}  \Bigr].
\end{align}
Then it is simple to check that the output probability follows $D_{s,\eta}(N,k)$ with $\eta=\sin^2(\theta/2)$.
Using the Hadamard transformation for the \textsf{RX} gate $e^{-iX \theta/2} = H e^{-i Z \theta/2} H$ and the \textsf{CNOT} gate $\mathrm{CNOT}_{ij} = H_j \mathrm{CZ}_{ij} H_j$ (where $i$ is the control qubit and $j$ is the target qubit), the circuit has initial Hadamard gates followed by the diagonal gate.
We demonstrate this circuit in Fig.~\ref{fig:lpn_to_iqp} (upper).

However, the transformed circuit still does not have the final Hadamard layer on the first $k$ qubits to be the IQP circuit. We thus use the circuit identity given in Fig.~\ref{fig:lpn_to_iqp} (lower). We note that the measurement outcomes $m$ follows the uniform distribution $p(m=0)=p(m=1)=1/2$.

Combining these results, we construct a circuit as follows. For some $k$ and $N \geq 2k+1$, our circuit has the initial state $\ket{0}^{\otimes N}$ followed by the Hadamard layer $H^{\otimes N}$.
We next apply $R_Z(\theta)$ to the $2k+1$-th qubit, and then for a given $s \in \{0,1\}^k$, apply the \textsf{CZ} gate between $i$-th and $2k+1$-th qubits if $s_i =1$.
After that, we apply the \textsf{CZ} gates between $i$-th and $k+i$-th qubits for all $1 \leq i \leq k$.
Finally, the last Hadamard later $H^{\otimes N}$  is applied and all qubits are measured in the computational basis.
For this circuit, one can check that the outcome distribution is $\Id_{k} \otimes P_{s, \eta} \otimes T_{n-(2k+1)}$, a product between three distributions: The first $k$ bits follows the uniform distribution, the next $k+1$ bits follow $P_{s, \eta}$, and the last $n-(2k+1)$ bits follow the trivial distribution $T_{n-(2k+1)}$.
After decomposing the \textsf{CZ} gate into the \textsf{RZ} and \textsf{IsingZZ} gates as $\mathrm{CZ}=e^{i \pi /4} e^{i \pi/4 Z \otimes Z} e^{ -i \pi/4 Z \otimes \Id}e^{ -i \pi/4 \Id \otimes Z}$, we see that the circuit is exactly the IQP circuit we considered in the main text with $2k$ two-qubit gates.

Thus under the assumption that Conjecture~\ref{conj:lpn_modified} is true, there is an IQP circuit with $2 n^{\Omega(1)} = n^{\Omega(1)}$ two-qubit gates (i.e., with underlying graph of $2n^{\Omega(1)}$ edges), whose output distribution is not efficiently learnable w.r.t. an evaluator.

We finally note that it is also possible to write down a \textit{single} IQP circuit in a two-dimensional lattice whose output distribution is not classically efficiently evaluated by a polynomial-size circuit (unless the polynomial hierarchy collapses) up to a constant multiplicative error~\cite{gao2017efficient}.

\section{Classical parent Hamiltonian of the output distribution} \label{app:classical_ham}
In the main text, we introduced a measure given by the summation of weights of the classical parent distribution to argue the classical hardness of learning the distribution.
In this section, we show the this measure can be computed when the probability distribution $p(x)$ is given, and discuss how this measure scales when we allow an approximate distribution.

\subsection{Algorithm for obtaining the classical parent Hamiltonian}
As we only consider binary outcomes, i.e. $x=(x_i)$ where $x_i\in\{0, 1\}$, we can always write $H(x) = -\sum_{S} J_S \prod_{i \in S} Z_{i}(x)$ where $S$ is the set of all subsets of indices $\{1, \cdots, N \}$ and $Z_i$ is a classical spin variable defined as $Z_i(x) = (-1)^{x_i}$.
For simplicity, we rewrite the sum using a binary representation of $S$ given as $y(S)=(y_i)$ where $y_i = 1$ if $i \in S$ or $0$ otherwise, which yields $H(x) = -\sum_{y=0}^{2^N-1} J(y) \mathbf{Z}_{y}(x)$ where $\mathbf{Z}_{y}(x) = (-1)^{x \cdot y}$.
As this is nothing but the Walsh-Hadamard transformation ($W\{f(x)\} = \sum_y f(x) (-1)^{x \cdot y}$), we can obtain $J(y)$ using the inverse transformation, i.e. $J(y) = W^{-1}\{-H(x)\}$, which is $(1/2^N) \sum_{x} \log p(x) (-1)^{x \cdot y}$ when $y \neq 0$ and $J(y=0) = -\log Z$.
The Gibbs state of this Hamiltonian \textit{exactly} represents the given distribution $p(x)$.
Our method here works when $p(x) \gneq 0$ for all $x$, which holds for the output distribution of the IQP circuits.

\subsection{Analytic expression of the classical parent Hamiltonian when $q=0.0$}
When $q=0.0$, we obtain
\begin{align}
    \braket{x|H^{\otimes N} e^{\sum_i \theta_i Z_i} | +}^{\otimes N} = \braket{x | e^{\sum_i i \theta_i X_i} | 0}^{\otimes N} = \prod_k \braket{x_k| \cos(\theta_k) \Id + i \sin(\theta_k)X | 0}.
\end{align} 
We thus have $p(x) = |\braket{x|H^{\otimes N} e^{\sum_i \theta_i Z_i} | +}^{\otimes N}|^2 = \prod_k f_{x_k}(\theta_k)$ where $f_{0}(\theta) = \cos^2(\theta)$ and $f_1(\theta) = \sin^2(\theta)$. 
It follows that the classical Hamiltonian $-H=\sum_i \alpha_i Z_i$ reproduces $p(x)$ for $\alpha_i = \log \tanh (\theta_i)$.

\subsection{Truncating Hamiltonian and approximating the distribution for a given target error}~\label{app:approx_hamiltonian}
In the main text, we used a Hamiltonian that exactly represents the given target distribution. 
However, one can also obtain a Hamiltonian that \emph{approximates} the target distribution by truncating small terms.
In fact, when there are lots of small terms in the Hamiltonian, we can ignore them when the target accuracy $\epsilon$ is given.

\begin{proposition}[Efficient approximate description of amplitudes]
For $0 < \epsilon <1$, suppose that we have $\tilde{H}(x)$ which satisfies $\max_x |H(x)-\tilde{H}(x)| < \epsilon$. Let $p(x) = e^{-H(x)}/Z$ and $\tilde{p}(x) = e^{-\tilde{H}(x)}/\widetilde{Z}$ where $Z = \sum_x e^{-H(x)}$ and $\widetilde{Z} = \sum_x e^{-\tilde{H}(x)}$. 
Then $|| p(x) - \tilde{p}(x)||_1 := \sum_x |p(x) - \tilde{p}(x)| \leq 2(e-1)\epsilon$.
\end{proposition}
\begin{proof}
\begin{align*}
    &\sum_x |p(x) - \tilde{p}(x)| = \sum_x |e^{-H(x)}/Z - e^{-\tilde{H}(x)}/\widetilde{Z}| \\
    &\quad \leq \sum_x |e^{-H(x)}/Z-e^{-\tilde{H}(x)}/Z|  + \sum_x|e^{-\tilde{H}(x)}/Z - e^{\tilde{H}(x)}/\widetilde{Z}| \\
    &\quad \leq \frac{1}{Z} \sum_x e^{-H(x)}|1-e^{-[\tilde{H}(x)-H(x)]}|  + \sum_x \frac{e^{-\tilde{H}(x)}}{\widetilde{Z}} \biggl| \frac{\widetilde{Z}}{Z} - 1\biggr| \\
    &\quad \leq \max_x |1-e^{-[\tilde{H}(x) - H(x)]}| + | \widetilde{Z}/Z - 1| \numberthis
\end{align*}
Furthermore, $| \widetilde{Z}/Z - 1| = |\widetilde{Z}-Z|/Z = |\sum_x e^{-\tilde{H}(x)} - e^{-H(x)}|/Z = |\sum_x e^{-H(x)} (e^{-[\tilde{H}(x) - H(x)]} - 1) |/Z \leq \max_x |1-e^{-[\tilde{H}(x) - H(x)]}|$.
Thus we have the bound $\sum_x |p(x) - \tilde{p}(x)| \leq 2 \max_x |1-e^{-[\tilde{H}(x) - H(x)]}|$. 
As $|1-e^{-y}|\leq (e-1)|y|$ for $|y| \leq 1$, we have the bound $2(e-1)\epsilon$.
\end{proof}

\begin{corollary}
For a probability distribution $p(x)= e^{-\sum_y J(y) \mathbf{Z}_y(x)}$ and $0 < \epsilon < 1$, we define 
\begin{align}
    J_{t}(y) = \begin{cases}
    J(y) & \text{ if } |J(y)| \geq \epsilon/2^N \\
    0 & \text{ otherwise.} \\
    \end{cases}
\end{align}
Then the probability distribution $p_t(x) \propto e^{f_t(x)}$ for $f_t(x) = W^{-1}\{J_t(x)\}$ satisfies $||p_t(x)-p(x)||_1 \leq 2(e-1)\epsilon$
\end{corollary}
\begin{proof}
$\max_x |f_t(x) - f(x)| \leq \max_x |\sum_y [J_t(y) - J(y)](-1)^{x \cdot y}| \leq \max_x \sum_y |J_t(y) - J(y)||(-1)^{x \cdot y}| \leq \epsilon/2^n \sum_y |(-1)^{x \cdot y}| = \epsilon$.
\end{proof}

\begin{figure}[h]
    \centering
    \includegraphics[width=0.75\textwidth]{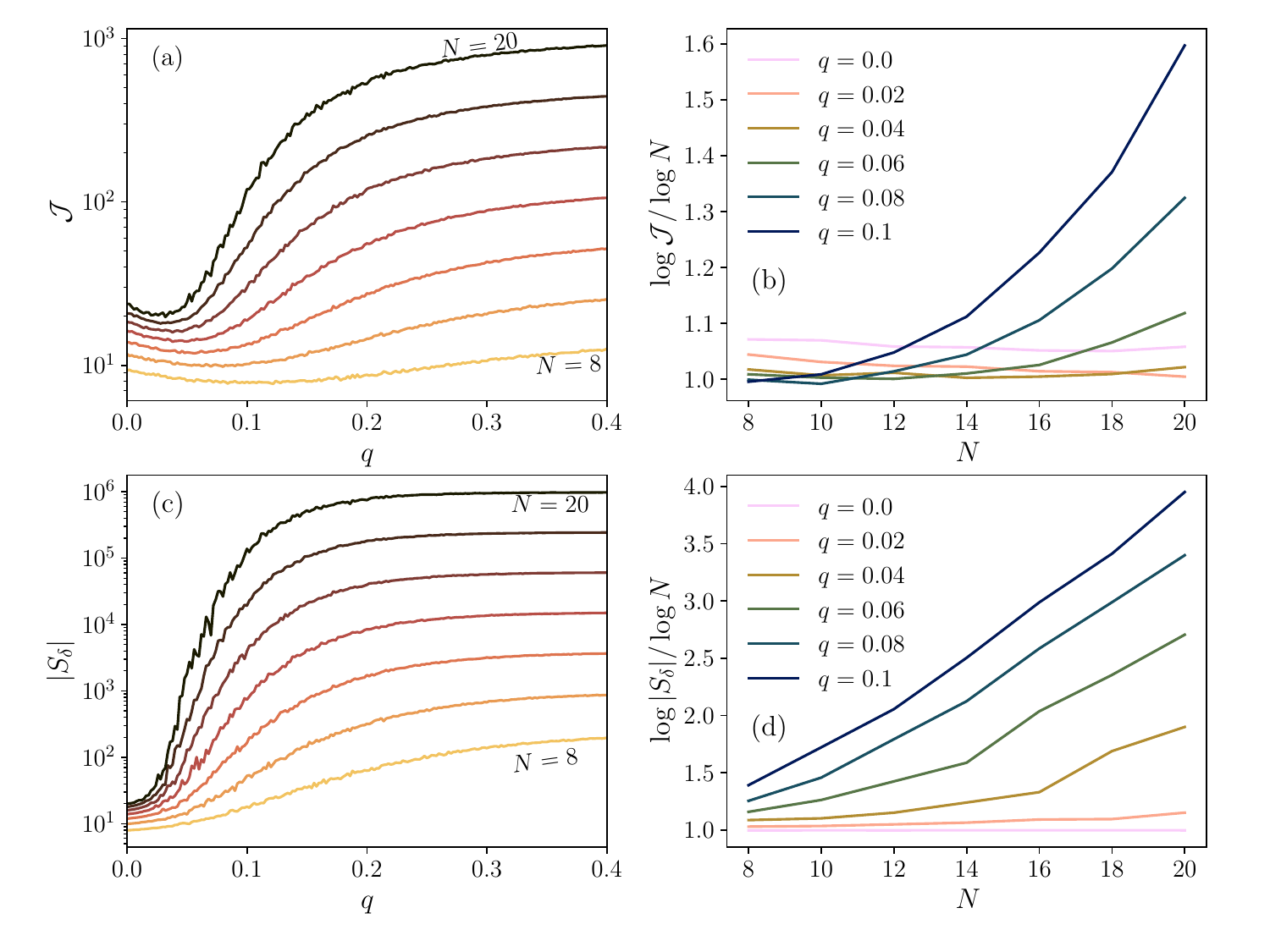}
    \caption{(a) The sum of the Hamiltonian coefficients $\mathcal{J}$ studied in the main text and (b) its scaling behavior.
    The remaining number of terms of the parent Hamiltonian $|S_\delta|$ for fixed target error $\epsilon=10^{-3}$ (c) as a function of $q$ and (d) as a function of $N$ with a normalization to observe the scaling behavior. From (b) and (d), we see that both measures increase super-polynomially with $N$ when $q \gtrsim 0.04$. }
    \label{fig:ham_scaling}
\end{figure}

One can also prove the inverse under a multiplicative approximation error.
\begin{proposition}
For given $p(x) > 0$, consider $\tilde{p}(x)$ that approximates $p(x)$ as $(1/c) p(x) < \tilde{p}(x) < c p(x)$ for $c>1$, then $|J(y) - \tilde{J}(y)| < \log c$ for all $y \neq 0$.
\end{proposition}
\begin{proof}
This follows from the inverse transformation $J(y \neq 0) = (1/2^N)\sum_x \log p(x) (-1)^{x \cdot y}$.
\end{proof}

However, multiplicative approximations are not practical in most cases. 
As reduction to the additive error is not straightforward,
we instead introduce a simple numerical technique:
For given $H(x)$ and $J(y)=W\{-H(x)\}$, we consider $p_\delta(x) \propto \exp[-W^{-1}\{J_\delta(y)\}]$ where $J_\delta (y)$ is a term after truncation given as
\begin{align}
    J_\delta (y) = \begin{cases}
        0 & \text{ if } |J(y)| < \delta \\
        J(y) & \text{otherwise}.
    \end{cases}
\end{align}
We then numerically find a maximal $\delta$ which satisfies $||p(x)-p_\delta(x)||_1 < \epsilon$ for a given target error $\epsilon$.
In practice, we use the following algorithm: First, we divide the range $[\min_k J(y), \max_k J(y)]$ evenly on a log scale $\{\delta_1,\delta_2,\cdots,\delta_n\}$. For each $\delta_k$, we compute $||p_{\delta}(x) - p(x)||_1$ and find a range $[\delta_i, \delta_{i+1}]$ such that $||p_{\delta_i}(x) - p(x)||_1 < \epsilon < ||p_{\delta_{i+1}}(x) - p(x)||_1$.  
We then binary search between $[\delta_i, \delta_{i+1}]$ to find the target $\delta$.

\subsection{Other signatures of phase transitions in classical parent Hamiltonians}~\label{app:phase_transition_signatures}
In this subsection, we give additional data supporting phase transitions in the learnability of the distribution based on our Hamiltonian measure.
First, we revisit $\mathcal{J}$ studied in the main text in Fig.~\ref{fig:ham_scaling}(a) and (b), which show this measure as a function of $q$ and its scaling behavior with $N$ for different $q$, respectively. Especially, we see that $\widetilde{J} = \log \mathcal{J} / \log N$ diverges as $N$ increases in Fig.~\ref{fig:ham_scaling}(b), which implies $\mathcal{J}$ increases super-polynomially.

We next study the number of remaining terms of the parent Hamiltonian obtained from the algorithm described in the previous section. 
After finding $\delta$ for a given target error $\epsilon$,
we compute the number of remaining terms of the Hamiltonian, which is given by $|S_\delta|$ where $S_\delta = \{y \neq 0| J_\delta (y) \neq 0\}$. 

For target error $\epsilon = 10^{-3}$, we plot $|S_\delta|$ in Fig.~\ref{fig:ham_scaling}(c) and (d).
Similar to the sum of the Hamiltonian $\sum_{y \neq 0} |J(y)|$, Fig.~\ref{fig:ham_scaling}(d) shows that the $|S_\delta|$ scales super-polynomially even when $q \lesssim 1.0 < q_{c_0}$. 
To be specific, $\log|S_\delta|/\log(N)$ increases faster than linear when $q=0.08$ and $q=0.10$, which is not expected if $|S_\delta|=O(N^\alpha)$ for some $\alpha>0$. 
In addition, Fig.~\ref{fig:ham_scaling}(c) shows a plateus for $q > q_{c_0}$, which we consider a property of the scrambled phase.

We additionally study each term of the parent Hamiltonian $J(y)$ in Fig.~\ref{fig:ham_weights}. Instead of listing all $J(y)$ (whose domain is $y\in{0,\cdots,2^N-1}$), we use an aggregated version of it given by $\sum_{w(y)=k} |J(y)|$ where $w(y)$ is the Hamming weight of $y$ ($\sum_{i=0}^{N-1} y_i$).
We plot these quantities obtained for $N=20$ in Fig.~\ref{fig:ham_weights}, which show three different behaviors depending on $q$. First, when $q$ is small, we see that it has only a single peak.
Next, for $0.02 \lesssim q \lesssim 0.06$, we observe that a slope is developed for small $k$.
The shape becomes a usual binomial distribution as $q$ further increases. As constant $|J(y)|$ yields $\sum_{w(y)=k} |J(y)| \propto \sum_{w(y)=k} 1 = {N \choose k}$, we regard this as a signature that $|J(y)|$ becomes uniform.
Moreover, the plot also shows that the magnitude of the peak increases as $q$ increases in this phase.

\begin{figure}
    \centering
    \includegraphics[width=0.9\textwidth]{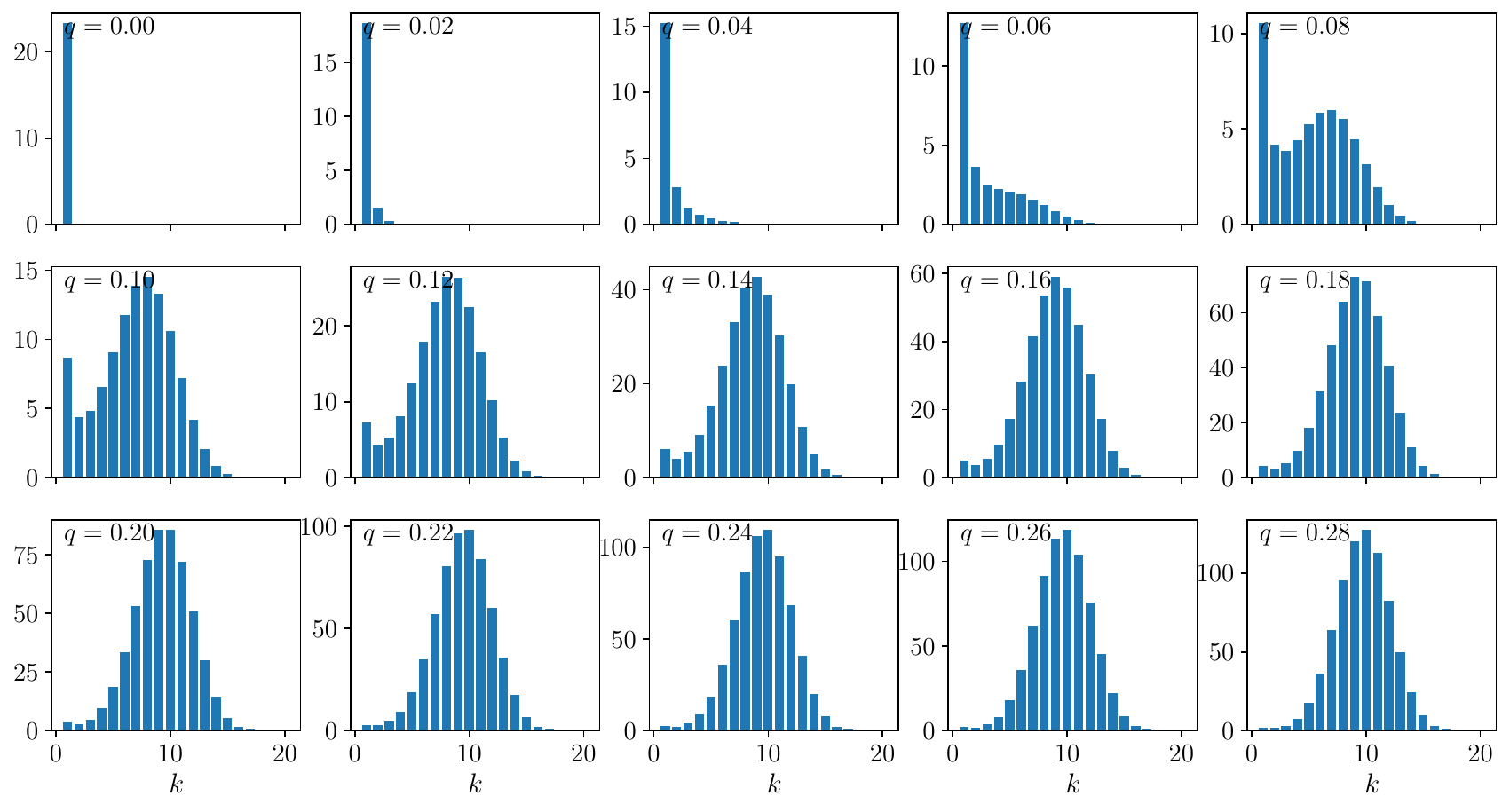}
    \caption{Averaged weights of the classical parent Hamiltonian for different $q \in [0.0, 0.02, \cdots, 0.28]$. For each value $k$ of the $x$-axis, we plot $\sum_{w(y)=k} |J(y)|$ where $w(y)$ is the Hamming weight of $y$. As we discussed in Sec.~\ref{app:classical_ham}, we see a single peak at $N-1$ when $q=0.0$. We also see that the plot resembles the binomial distribution when $q$ is large enough ($q \gtrsim 0.14$), which is a signature that the Hamiltonian coefficients become uniform $|J(y)| \approx C$ (see text for details). In between, we see that another structured behavior where both sides of the distribution grow.}
    \label{fig:ham_weights}
\end{figure}

\section{Details of machine learning setup}~\label{app:ml_setup}

We describe details of our machine learning setup in this Section. 

\subsection{Modelling the distribution with the energy based model}
The energy based model uses a neural network of a single real-number output as the energy function, i.e. $p_{W}(x) = e^{f_W(x)}/Z$ where $p_W(x)$ is the probability the model describes, $f_W(x)$ is the output of the network, and $Z = \sum_{x} f_{W}(x)$ is the partition function. Here $W$ is a one-dimensional vector of all parameters of the network.
One can train this model using a gradient-descent type of algorithm by minimizing the cross entropy $L[p_W(x)] = -\sum_x p(x) \log p_{W}(x) = C + H(p)$ where $p(x)$ is the target distribution and $H(p) = -\sum_x p(x)\log p(x)$ is the Shannon entropy of $p(x)$. 
Note also that we use $p(x)$ as the reference distribution for the KL divergence here, which is opposite to the KL divergence we have shown in Fig.~4 in the main text (where we have used $p_W(x)$ as the reference distribution).
Evaluating $L[p_W(x)]$ is not straightforward for the energy-based model in practice, as the partition function $Z$ is hard to be evaluated. Instead, it is more feasible to estimate the gradient of the loss function, which is given by
\begin{align}
    -g &= -\nabla_W L[p_W(x)] = -\nabla_W \Bigl[ -\sum_x p(x) \log p_{W}(x) \Bigr] = \sum_x p(x) \nabla_W [\log p_W(x) - \log Z] \\
    &= \sum_x p(x) \nabla_W f_W(x) - \sum_{x'} \frac{[\nabla_W f_W(x')] e^{f_W(x')}}{Z} \\
    &\approx \bigl\langle \nabla_W f_W(x) \bigr\rangle_{x \sim p(x)} -\bigl\langle \nabla_W f_W(x) \bigr\rangle_{x \sim p_W(x)},
\end{align}
where $\braket{\cdot}_{x \sim p(x)}$ is an average over samples from $p(x)$. Then we can train the network using usual first order optimization methods such as Adam~\cite{kingma2014adam}.

\begin{figure}[t]
    \centering
    \includegraphics[width=0.9 \textwidth]{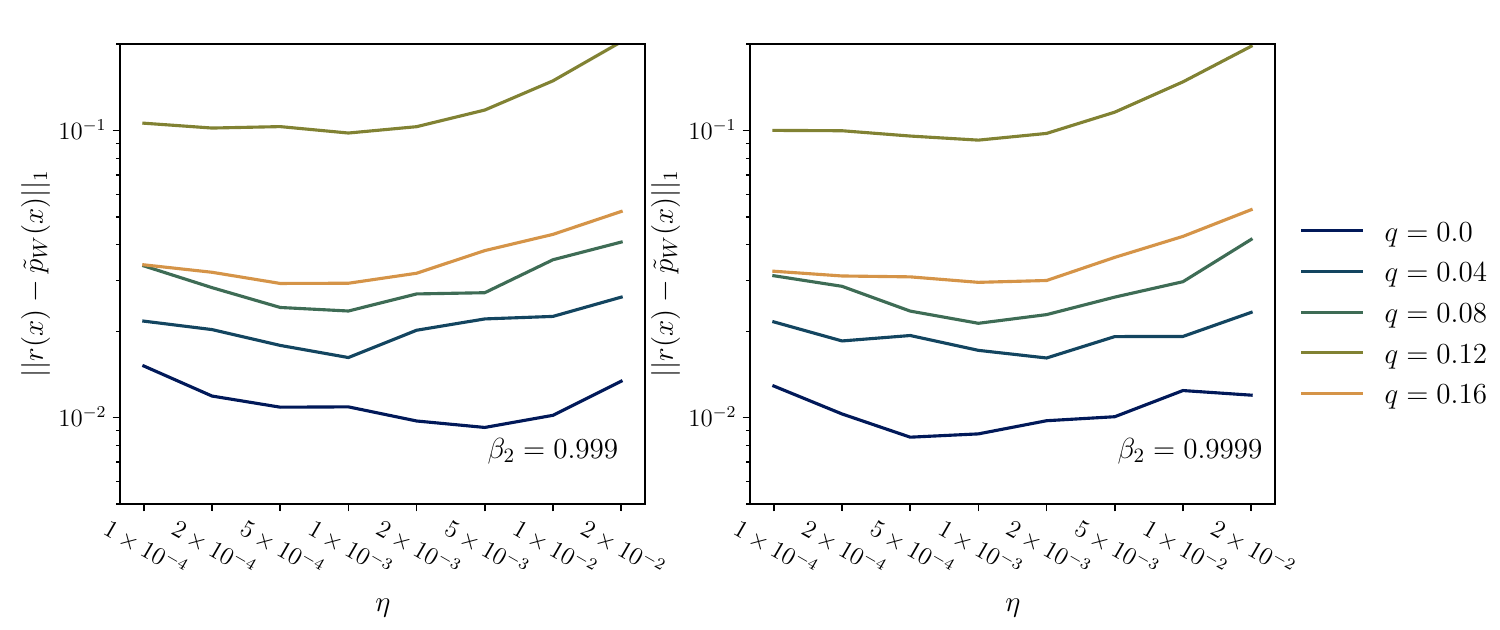}
    \caption{Converged $L1$ distance between the converged distribution from an energy based model $p_W(x)$ and the distribution from the IQP circuit $p(x)$. 
    For each $q$, we randomly pick a single IQP instance and trained the energy-based model for different hyperparameters $\eta$ (learning rate) and $\beta_2$ (average rate for the Fisher information matrix). }
    \label{fig:learn_iqp_hyper_adam}
\end{figure}

In general, one samples from the distribution $p_W(x)$ with Monte-Carlo methods such as Metropolis-Hastings algorithm. In our simulation, however, we compute $p_W(x)$ exactly by computing $f_W(x)$ over all possible $x$ and exactly sample from the computed distribution. This is feasible up to $N=20$ qubits (the maximum number of qubits we simulated) by maintaining all $2^N$ values. Thus we always generate unbiased samples regardless of the shape of the distribution (e.g. even for a spin-glass type distribution that Monte-Carlo methods may fail).

For each epoch, we sample $x=x_1,\cdots,x_S$ from $p(x)$ as well as $p_W(x)$ where $M$ is the min-batch size and estimate $F$ and $g$ from them.
As we generate samples from $p(x)$ for each epoch, usually bigger sample size $S$ performs better (see also e.g. Ref.~\cite{pfau2020ab,park2022expressive}). In our simulation, we choose $S=1024$ which is the maximum size (among the perfect power of $2$) we can simulate with the given computational resource.

\subsection{Neural network architecture}
As our target distribution (the output distribution of a circuit) does not have a particular symmetry, we use a network with fully connected layers for $f_W(x)$.
We use a fully-connected network with two hidden layers each has $M=\alpha N$ hidden units. 
The visible and the first hidden layers have bias terms whereas the last layer does not. Thus the number of parameters from the first, second, and the last layers are $N\times M+M$, $M\times M + M$, and $M$, respectively. 
We used $\alpha=30$ for the results in the main text, and the network has total $60,240$ parameters for $N=8$ and $373,800$ for $N=20$.

\subsection{Hyperparameters}
The hyperparameters for Adam optimizer are the learning rate ($\eta$), the momentum for the gradient ($\beta_1$), and the momentum for component-wise gradient norm ($\beta_2$). We use the default value for $\beta_1=0.9$ and grid search $\eta$ for $\beta_2=[0.999, 0.9999]$.
For a single circuit realization for given $q$, the converged $L1$ distance between the target distribution $p(x)$ and $p_W(x)$ is shown in Fig.~\ref{fig:learn_iqp_hyper_adam}.
Based on the results, we choose $\eta = 10^{-3}$ and $\beta_2 = 0.999$, which is well-performing for all tested values of $q$, to obtain results in the main text.
We still expect similar quantitative results from other choices of hyperparamters.

\section{Results from the natural gradient descent}\label{app:nat_grad}

\begin{figure}
    \centering
    \includegraphics[width=0.45\linewidth]{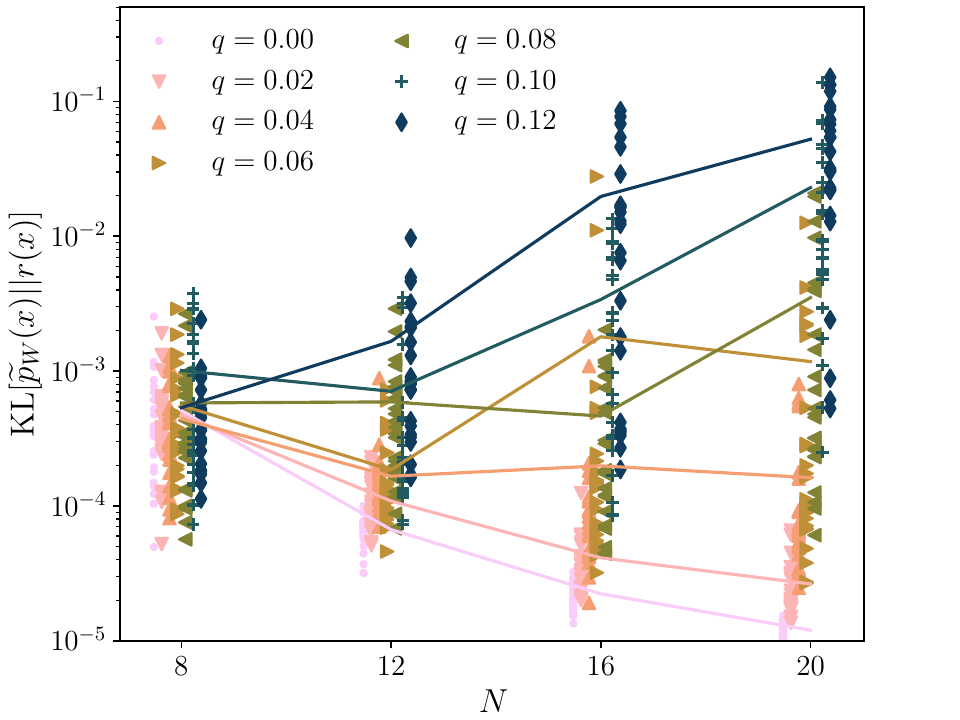}
    \caption{The KL divergence between the output distribution of the IQP circuits $p(x)$ and that from the energy based model after training $\tilde{p}_W(x)$ with the NGD. 
    Other settings are the same as Fig.~4 in the main text.}
    \label{fig:ml_converged_kl_ngd}
\end{figure}

\subsection{Natural gradient descent optimizer}
The natural gradient descent (NGD)~\cite{amari1998natural} is a second order method for optimizing a neural network.
Albeit it requires a significantly larger amount of computational resources, it may give better convergence properties for some dataset~\cite{park2020geometry,pfau2020ab}.
For the natural gradient descent, we additionally compute the Fisher information matrix given as $F = (F_{ij})$ where
\begin{align}
F_{ij} &= \sum_x p_W(x) [\partial_{W_i}\log p_W(x)][\partial_{W_j}\log p_W(x)] - \sum_x p(x) [\partial_{W_i}\log p_W(x)] \sum_x p(x) [\partial_{W_j}\log p_W(x)] \\
       &\approx \bigl\langle \partial_{W_i} f_{W}(x) \partial_{W_j} f_{W}(x) \bigr\rangle_{x \sim p_W(x)} - \bigl\langle \partial_{W_i} f_{W}(x) \bigr\rangle_{x \sim p_W(x)}\bigl\langle \partial_{W_j} f_{W}(x) \bigr\rangle_{x \sim p_W(x)}.
\end{align}
Then the parameter of the network is updated as $W_{t+1} = W_t - \eta F^{-1}_t g_t$ where $W_t$ is the weight at epoch $t$, $F_t$ and $g_t$ are moving averages of the Fisher matrix $F$ and the gradient $g$. We use $F_t = F_{t-1} + \beta_2 F$ and $g_t = g_{t-1} + \beta_1 g$ where $F$ and $g$ are Fisher matrix and the gradient estimated at each epoch, respectively.
We note that we solve the linear equation $F_t v = g_t$ at each step, which requires $O(|W|^{2-3})$ number of operations (depending on the condition number of $F$) where $|W|$ is the number of parameters in the network. 
As the first order optimization algorithms (e.g. Adam) only require $O(|W|)$ operations for each step, we see that the NGD requires a significantly larger amount of computational resources.

\subsection{Neural network architecture}
Because of the computational overhead mentioned above, we use a relatively smaller network (compared to the case of Adam).
Our network has one hidden layers which has $10N$ hidden units where $N$ is the number of qubits. Only the visible layer has a bias term and the total number of the parameters of the network is $10N^2 + 11N$. We note that the number of parameters in this case is still an order of magnitude larger than that of the IQP circuit itself.

\subsection{Results}
As for the Adam optimizer, we plot the converged distance between $p(x)$ and $\tilde{p}_W(x)$ for $24$ different circuit realizations after hyperparameter optimization in Fig.~\ref{fig:ml_converged_kl_ngd}.
We see that the errors increase exponentially with $N$ for $0.06 \lesssim q < q_{c_0}$ as in the results from Adam. 
We also observe that the converged KL divergence decreases with $N$ for $q=0.00$ and $q=0.02$, which is different from what we observed with Adam optimizer. As the converged values are indeed smaller (the averaged converged KL divergence is $< 2\times 10^{-5}$ when $N=20$ and $q=0.0$) than when we used Adam (which was $>10^{-4}$) albeit its smaller network size, we conclude that the distributions when $q < q_{c_1}$ are better handled by the NGD than Adam.

\end{document}